\documentclass{IOS-Book-Article}
\usepackage{graphicx}
\usepackage{multirow}
\usepackage{amsthm}
\usepackage{url}

\newtheorem{thrm}{Theorem}
\newtheorem{lmm}{Lemma}
\newtheorem{coro}{Corollary}

\theoremstyle{definition}
\newtheorem{defn}{Definition}
\newtheorem{example}{Example}

\newcommand{\AR}[0]{\mathit{AR}}
\newcommand{\ATT}[0]{\mathit{ATT}}

\newcommand{\Var}[0]{\mathit{Var}}
\newcommand{\Lin}[0]{\mathrm{in}}
\newcommand{\Lout}[0]{\mathrm{out}}
\newcommand{\Lundec}[0]{\mathrm{undec}}

\newcommand{\evalV}[2]{[ \! [ #1 ] \! ]_{#2}}

\newcommand{\LV}[1]{\mathrm{LV}(#1)}
\newcommand{\RE}[1]{\mathrm{RE}(#1)}
\newcommand{\REF}[1]{\mathrm{ref}(#1)}

\newcommand{\VArg}[0]{\mathrm{VA}}

\begin{document}
\begin{frontmatter}
\title{Compiling Arguments in an Argumentation Framework into Three-valued Logical Expressions}
\runningtitle{Compiling Arguments into Three-valued Logical Expressions}

\author[A]{\fnms{Sosuke} \snm{Moriguchi}%
\thanks{Corresponding Author: Contract Assistant, Kwansei Gakuin University, 2-1 Gakuen, Sanda, Hyogo, 669-1337 Japan;
E-mail:chiguri@acm.org.}}
and
\author[A]{\fnms{Kazuko} \snm{Takahashi}}

\runningauthor{S. Moriguchi and K. Takahashi}
\address[A]{Kwansei Gakuin University, Japan}
\begin{abstract}
In this paper, we propose a new method for computing general allocators directly from completeness conditions.
A general allocator is an abstraction of all complete labelings for an argumentation framework.
Any complete labeling is obtained from a general allocator by assigning logical constants to variables.
We proved the existence of the general allocators in our previous work.
However, the construction requires us to enumerate all complete labelings for the framework, which makes the computation prohibitively slow.
The method proposed in this paper enables us to compute general allocators without enumerating complete labelings.
It also provides the solutions of local allocation that yield semantics for subsets of the framework.
We demonstrate two applications of general allocators, stability, and a new concept for frameworks, termed arity.
Moreover, the method, including local allocation, is applicable to broad extensions of frameworks, such as argumentation frameworks with set-attacks, bipolar argumentation frameworks, and abstract dialectical frameworks.
\end{abstract}


\end{frontmatter}

\section{Introduction}
%%% AFが使われる場面と利便性？

Dung's abstract argumentation framework (AF)~\cite{Dung95} focuses on the relations between attacking/attacked arguments.
Its semantics, consistent with the interpretation of the framework, are denoted as extensions (sets of the arguments) and labelings (labeling functions from the arguments to labels `in,' `out,' and `undec').
Such semantics shows the dependency between directly and indirectly related arguments.
Table~\ref{tbl:example1} shows all of the complete labelings (one of the types of semantics) for the framework $(\{ 1, 2, 3, 4 \}, \{ (1, 2), (2, 1), (1, 3), (2, 3), (3, 4) \})$.
The completeness of the labelings justifies the rows, but says nothing about the columns.
Furthermore, it is difficult to clarify the relationships between such arguments because this behavior is not demonstrated explicitly.

%%% labeling、この書式で書く？
\begin{table}[tbp]
\caption{The list of complete labelings for $(\{ 1, 2, 3, 4 \}, \{ (1, 2), (2, 1), (1, 3), (2, 3), (3, 4) \})$.} \label{tbl:example1}
\begin{tabular}{c|cccc}
                & \multicolumn{4}{c}{Arguments} \\
Labeling        & 1     & 2     & 3     & 4 \\ \hline
$\mathcal{L}_1$ & undec & undec & undec & undec \\
$\mathcal{L}_2$ & in    & out   & out   & in \\
$\mathcal{L}_3$ & out   & in    & out   & in \\
\end{tabular}
\end{table}

We proposed a novel allocation method in \cite{Moriguchi18}.
The allocation method reframes the notion of the acceptability of the arguments in the framework and the possible behavior of each argument in the framework.
In the allocation method, we use an allocator, a function from arguments to three-valued logical expressions.
Allocated expressions show how the acceptance of arguments changes in the framework (as the columns in Table~\ref{tbl:example1}).
A general allocator is an allocator that abstracts all possible labeling functions for the framework.
It enables us to compare two arguments based on the expressions assigned to the arguments.
Constructing general allocators from the framework can be seen as compiling knowledge in the framework into expressions.
However, there are problems with the construction of general allocators presented in \cite{Moriguchi18}, in terms of both time and memory complexity.
The construction requires all complete labelings.
We fuse two complete labelings together and build up a general allocator.
The number of complete labelings is often greater than the size of the framework.

In this paper, we provide another algorithm for constructing general allocators.
The proposed algorithm solves the equations for the completeness conditions directly.
Hence, according to this algorithm, it is not necessary to compute the complete labelings.

The contributions of this paper are as follows.
\begin{itemize}
 \item Proposal of a new algorithm to build a general allocator:
 The algorithm is simple enough to implement easily.
 We implement a prototype solver based on the algorithm.
 
 \item Application of the algorithm to local allocation: The algorithm is easily extended to local allocation.
 Local allocation provides semantics for part of the framework by abstracting effects from outside of that part.
 Previously, we only discussed the completeness of local allocation.
 In this paper, we define general local allocators and present a construction algorithm.

 \item Applications of the general allocators:
 We demonstrate two applications of general allocators.
 One is for the derivation of stable labelings from a general allocator, and the other is for a new concept, termed arity.
 %% arityの簡単な説明ほしくない？

 \item Allocation method for abstract dialectical frameworks:
 The algorithm is also applicable to some extensions of AFs, such as set-attacks, bipolar argumentation frameworks, and abstract dialectical frameworks.
 This is because the algorithm requires very few conditions to compute the general (local) allocators.
\end{itemize}

The rest of the paper is organized as follows.
We review the basic notions of AFs and labeling in Section~\ref{sec:AF}, and the allocation method proposed in \cite{Moriguchi18} in Section~\ref{sec:allocation}.
In Section~\ref{sec:equation}, we propose the algorithm to compute general allocators for the given framework.
In Section~\ref{sec:localallocation}, we show a local allocation method and extension of the algorithm.
In Section~\ref{sec:discussion}, we discuss stability in the allocation method and numbers of variables in general allocators.
In Section~\ref{sec:relatedwork}, we compare the method to some methods for calculating labelings and show applicability of the method for extensions of the frameworks.
Finally, we conclude this paper in Section~\ref{sec:conclusion}.

\section{Argumentation Framework} \label{sec:AF}
We will begin by defining AFs, using the terms employed to express Dung's concepts in \cite{Dung95}.

\begin{defn}
An argumentation framework (AF) is a pair of a set of arguments and their attack relations (i.e., binary relation on arguments).
We use $\langle \AR, \ATT \rangle$ to denote an AF.
\end{defn}

The definition of labeling (with regard to semantics) offered by Caminada~\cite{Caminada06} is as follows.

\begin{defn}
Labeling in terms of frameworks is a function from arguments to labels, i.e., $L : \AR \rightarrow \{ \Lin, \: \Lundec, \: \Lout \}$.
Labeling $L$ is complete iff the following conditions are satisfied.
\begin{itemize}
 \item $L(A) = \Lin$ iff $L(A') = \Lout$ for all arguments $A'$ such that $(A', A) \in \ATT$.
 \item $L(A) = \Lout$ iff there exists an argument $A'$ such that $(A', A) \in \ATT$ and $L(A') = \Lin$.
 \item $L(A) = \Lundec$ iff there exists an argument $A'$ such that $(A', A) \in \ATT$ and $L(A') = \Lundec$  and there are no arguments $A''$ such that $(A'', A) \in \ATT$ and $L(A'') = \Lin$.
\end{itemize}
The complete labeling $L$ is \emph{grounded} iff $\{ A | L(A) = \Lin \}$ is the smallest in those complete labelings w.r.t. set inclusion.
We use $L_g$ for grounded labeling.
The complete labeling $L$ is \emph{stable} iff $\{ A | L(A) = \Lundec \}$ is empty.
\end{defn}
The arguments labeled `$\Lin$' are considered to be accepted.

\begin{example} \label{example:labeling}
In Figure~\ref{fig:AFexample1}, we illustrate two AFs, (a) and (b).
The left framework (a) is an acyclic framework, with only one complete labeling, $L(1) = L(3) = L(5) = \Lin$ and $L(2) = L(4) = \Lout$.
As the complete labeling is unique, it also constitutes a grounded labeling.

The right framework (b) incorporates a cyclic part between $1$ and $2$.
There are three complete labelings, $L_1$, $L_2$ and $L_3$:
\begin{itemize}
 \item $L_1(A) = \Lundec$ for all arguments $A$. This is a grounded labeling.
 \item $L_2(1) = L_2(4) = \Lin$ and $L_2(2) = L_2(3) = \Lout$.
 \item $L_3(2) = L_3(4) = \Lin$ and $L_3(1) = L_3(3) = \Lout$.
\end{itemize}
\end{example}

\begin{figure}[tbp]
 \begin{center}
 \includegraphics[height=2cm]{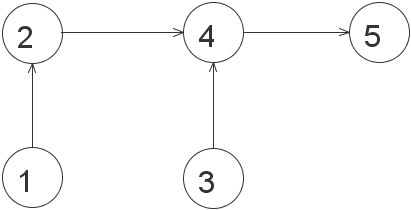}\hspace{2cm}
 \includegraphics[height=2cm]{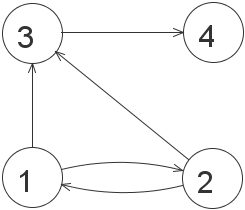}
 \end{center}
 \caption{Two examples of argumentation frameworks. (a) Left: acyclic graph, and (b) right: cyclic graph.}
 \label{fig:AFexample1}
\end{figure}

From these labelings, we can see that arguments $1$ and $2$ always have opposite labels, and $3$ and $4$ are labeled $\Lundec$ only when $1$ and $2$ are labeled $\Lundec$\footnote{Here, ``opposite'' means the negation in three-valued logic. `$\Lin$' and `$\Lout$' are opposite to each other, and `$\Lundec$' is opposite to itself.}.
However, each labeling does not, by itself, imply such an observation.

\section{Allocation Method} \label{sec:allocation}
Here, we explain the allocation method proposed in \cite{Moriguchi18}.
Before introducing the method, we explain the relevant three-valued logical expressions.

\subsection{Three-valued Logical Expression} \label{sec:three-valued}
Here, we define the three-valued logical expressions (henceforth, expressions for short) as follows:
$$ \begin{array}{l@{\quad}r@{\quad}l}
p & \mbox{::=} & T \quad | \quad F \quad | \quad U \quad | \quad x \quad | \quad  \neg p \quad | \quad p \land p \quad | \quad p \lor p \\
\end{array}$$
where $x$ is a variable (an element of $\Var$) and $T$, $F$, and $U$ are constants denoting true, false, and undecided (the middle value), respectively.
The operator precedence of these connectives is as usual; $\neg$, $\land$ and $\lor$.
In this section we write $\land$ explicitly, but in the following sections, we omit $\land$ for brevity.

We define the evaluation of the expressions under valuation for the variables $v : \Var \rightarrow \{ T, F, U \}$, as Kleene's three-valued logic~\cite{Kleene52}.
\begin{center}
$\evalV{T}{v} = T, \quad \evalV{F}{v} = F, \quad \evalV{U}{v} = U, \quad \evalV{x}{v} = v(x).$ \\
\begin{tabular}{|c|c|c|c|c|} \hline
\multicolumn{2}{|c|}{\multirow{2}{*}{$\evalV{p \land q}{v}$}} & \multicolumn{3}{|c|}{$\evalV{q}{v}$} \\ \cline{3-5}
\multicolumn{2}{|c|}{}                & $T$ & $U$ & $F$ \\ \hline
\multirow{3}{*}{$\evalV{p}{v}$} & $T$ & $T$ & $U$ & $F$ \\ \cline{2-5}
                                & $U$ & $U$ & $U$ & $F$ \\ \cline{2-5}
                                & $F$ & $F$ & $F$ & $F$ \\ \hline
\end{tabular} \quad
\begin{tabular}{|c|c|c|c|c|} \hline
\multicolumn{2}{|c|}{\multirow{2}{*}{$\evalV{p \lor q}{v}$}} & \multicolumn{3}{|c|}{$\evalV{q}{v}$} \\ \cline{3-5}
\multicolumn{2}{|c|}{}                & $T$ & $U$ & $F$ \\ \hline
\multirow{3}{*}{$\evalV{p}{v}$} & $T$ & $T$ & $T$ & $T$ \\ \cline{2-5}
                                & $U$ & $T$ & $U$ & $U$ \\ \cline{2-5}
                                & $F$ & $T$ & $U$ & $F$ \\ \hline
\end{tabular} \quad
\begin{tabular}{|c|c|c|} \hline
\multicolumn{3}{|c|}{$\evalV{\neg p}{v}$} \\ \hline
\multirow{3}{*}{$\evalV{p}{v}$} & $T$ & $F$ \\ \cline{2-3}
                                & $U$ & $U$ \\ \cline{2-3}
                                & $F$ & $T$ \\ \hline
\end{tabular}
\end{center}

We also define the equivalence between expressions as $p \equiv q \Leftrightarrow \forall v, \evalV{p}{v} = \evalV{q}{v}$.
This equivalence relation is clearly reflexive, symmetric, and transitive.

Some well-known equivalences in binary logic can be also proven in this logic.
\begin{lmm} \label{lmm:three-valued-equivalence-primitive}
For any expressions $p$, $q$ and $r$, the following equivalences hold.
$$\begin{array}{c}
\neg \neg p \equiv p   \qquad   \neg T \equiv F   \qquad   \neg F \equiv T \\ %% about negation
p \land p \equiv p   \qquad   p \land q \equiv q \land p   \qquad   T \land p \equiv p   \qquad   F \land p \equiv F \\ %% about conjunction   associativity is omitted
p \lor  p \equiv p   \qquad   p \lor  q \equiv q \lor  p   \qquad   T \lor  p \equiv T   \qquad   F \lor  p \equiv p \\ %% about disjunction
p \land (q \lor r) \equiv p \land q \lor p \land r \quad \neg (p \land q) \equiv \neg p \lor \neg q \quad \neg (p \lor q) \equiv \neg p \land \neg q \\ % combination of connectives (distributive, de Morgan's law)
\\
\begin{array}{c}
p \equiv q \\ \hline
\neg p \equiv \neg q \\
\end{array}
\qquad
\begin{array}{c}
p \equiv q \\ \hline
p \land r \equiv q \land r \\
\end{array}
\qquad
\begin{array}{c}
p \equiv q \\ \hline
p \lor r \equiv q \lor r \\
\end{array} %% about structural rules
\end{array}$$
\end{lmm}

However, some equivalences are not satisfied in three-valued logic.
\begin{lmm} \label{lmm:three-valued-not-equivalence}
Assume that an expression $p$ is not equivalent to either $T$ or $F$.
Then, $p \lor \neg p \not\equiv T$ and $p \land \neg p \not\equiv F$.
\end{lmm}

Let us consider more general properties of three-valued logic.
In contrast to binary logic, logical constants $T$ and $F$ play unique roles in logical expressions.

\begin{lmm} \label{lmm:undecvaluation}
Let the valuation $v_U$ be defined as $v_U(x) = U$ for every variable $x$.
If an expression $p$ does not include any occurrence of $T$ and $F$, $\evalV{p}{v_U} = U$, i.e., $p \not\equiv T$ and $p \not\equiv F$.
\end{lmm}
This follows straightforwardly by induction on $p$.

\begin{lmm} \label{lmm:expressionpattern}
Any expression $p$ satisfies exactly one of the following conditions.
\begin{enumerate}
 \item $p$ is equivalent to $T$.
 \item $p$ is equivalent to $F$.
 \item There is an expression $p'$ equivalent to $p$ such that $p'$ does not contain any occurrences of $T$ or $F$.
\end{enumerate}
\end{lmm}
\begin{proof}
This is proven by induction on the expression $p$.

When $p$ is $T$ (or $F$), $p$ satisfies the first (the second) condition.
When $p$ is $U$ or a variable, it is clear that $p$ only satisfies the third condition.

When $p = p_1 \land p_2$, we apply $p_1$ to the induction hypothesis.
1) If $p_1$ is equivalent to $T$, then $p$ is equivalent to $p_2$. According to induction hypothesis for $p_2$, $p$ satisfies the proposition.
2) If $p_1$ is equivalent to $F$, then $p$ is equivalent to $F$, i.e., $p$ satisfies the second condition.
3) Otherwise, $p_1$ is equivalent to $p_1'$, which do not include any occurrences of $T$ or $F$.
By applying $p_2$ to the induction hypothesis, we can split 3 cases:
3-1) If $p_2$ is equivalent to $T$, then $p$ is equivalent to $p_1'$, which satisfies the third condition.
3-2) If $p_2$ is equivalent to $F$, then $p$ is equivalent to $F$, i.e., $p$ satisfies the second condition.
3-3) Otherwise, $p_2$ is equivalent to $p_2'$, which do not include any occurrences of $T$ or $F$.
In this case, $p$ is equivalent to $p_1' \land p_2'$, which do not include any occurrences of $T$ or $F$.
From Lemma~\ref{lmm:undecvaluation}, $p_1' \land p_2'$ is not equivalent to $T$ or $F$, so $p$ only satisfies the third condition.

The case such that $p = p_1 \lor p_2$ or $p = \neg p_1$ is similar to the case of $p_1 \land p_2$.
\end{proof}

For example, $x_1 \land (x_2 \lor T)$ is equivalent to $x_1$ (satisfying the third condition), but not equivalent to $T$ nor $F$.
This is proven by the following process:
The subexpression $x_2 \lor T$ is equivalent to $T$, and $x_1 \land T$ is equivalent to $x_1$.
Such recursive process is the decision procedure whether the given expression is equivalent to $T$ or $F$, or not equivalent to both.
Simpler version of the decision procedure is derived from the following corollary.

\begin{coro} \label{coro:undecvaluation}
An expression $p$ is not equivalent to both $T$ or $F$ iff $\evalV{p}{v_U} = U$.
\end{coro}

From Lemma~\ref{lmm:expressionpattern} and Corollary~\ref{coro:undecvaluation}, we can classify an expression by evaluating it with $v_U$;
the expression is equivalent to $T$ if the result is $T$, $F$ if the result is $F$, and neither $T$ nor $F$ otherwise (i.e., the result is $U$).
For example, Lemma~\ref{lmm:three-valued-not-equivalence} can be proven by this procedure.

We apply substitution to an expression and a valuation, respectively, replacing the variables with other expressions.
\begin{defn}
$p[p'/x]$ for expressions $p$ and $p'$, and variable $x$ are defined as follows.
\begin{itemize}
 \item If $p$ is a logical constant, $p[p'/x] = p$.
 \item $x[p'/x] = p'$ and $y[p'/x] = y$ if $y \neq x$.
 \item $(\neg p)[p'/x] = \neg (p[p'/x])$, $(p \land q)[p'/x] = p[p'/x] \land q[p'/x]$, $(p \lor q)[p'/x] = p[p'/x] \lor q[p'/x]$.
\end{itemize}
We also use similar notation $v[p'/x]$ for the valuation $v$, expression $p'$, and variable $x$ to denote $v[p'/x](x) = \evalV{p'}{v}$ and $v[p'/x](y) = v(y)$ if $x \neq y$.
\end{defn}

\begin{lmm}\label{lmm:substitution}
Let $v$ be a valuation, $p$ and $p'$ be expressions, and $x$ be a variable.
Then, $\evalV{p[p'/x]}{v} = \evalV{p}{v[p'/x]}$.
\end{lmm}
\begin{proof}
The proof follows by induction on $p$.
$p$ is not a variable; it is proven directly from the induction hypothesis (or straightforward in the case of logical constants).
If $p = x$, $\evalV{p[p'/x]}{v} = \evalV{x[p'/x]}{v} = \evalV{p'}{v} = \evalV{x}{v[p'/x]} = \evalV{p}{v[p'/x]}$.
If $p \neq x$, $\evalV{p[p'/x]}{v} = \evalV{p}{v} = v(p) = v[p'/x](p) = \evalV{p}{v[p'/x]} = \evalV{p}{v[p'/x]}$.
Therefore, $\evalV{p[p'/x]}{v} = \evalV{p}{v[p'/x]}$.
\end{proof}

\begin{thrm}\label{thrm:subequiv}
Let $p_1$, $p_2$ and $p'$ be expressions, and $x$ be a variable.
If $p_1 \equiv p_2$, then $p_1[p'/x] \equiv p_2[p'/x]$.
\end{thrm}
\begin{proof}
For any valuation $v$, from Lemma~\ref{lmm:substitution}, $\evalV{p_1[p'/x]}{v} = \evalV{p_1}{v[p'/x]}$ and $\evalV{p_2[p'/x]}{v} = \evalV{p_2}{v[p'/x]}$.
Since $p_1 \equiv p_2$, $\evalV{p_1}{v[p'/x]} = \evalV{p_2}{v[p'/x]}$.
Therefore, $\evalV{p_1[p'/x]}{v} = \evalV{p_2[p'/x]}{v}$, and thus, $p_1[p'/x] \equiv p_2[p'/x]$.
\end{proof}

%\begin{lmm} %%% 結局これも使わなそう:やっぱりimplementation？
%For any expressions $p$ and $q$, the folowing equivalences hold.
%\begin{eqnarray}
%p (U \lor \neg p)        & \equiv & U p \lor p \neg p \equiv U p \label{eqn:a1} \\
%p (U \lor \neg p \lor q) & \equiv & p (U \lor q) \label{eqn:a2} \\
%p \lor p q               & \equiv & p \label{eqn:b} \\
%p (p \lor q)             & \equiv & p \label{eqn:c} \\
%p (U \lor \neg p q)      & \equiv & U p \label{eqn:a3}
%\end{eqnarray}
%\end{lmm}
%\begin{proof} %%% TODO 改行おかしいからなんか入れる？
%\begin{enumerate}
% \item $p (U \lor \neg p) \equiv U p \lor p \neg p \equiv U p \lor p p \neg p \equiv p (U \lor p \neg p) \equiv U p$
% \item $p (U \lor \neg p \lor q) \equiv p (U \lor \neg p) \lor p q \equiv U p \lor p q \equiv p (U \lor q)$ (from equation (\ref{eqn:a1}))
% \item $p \lor p q \equiv T p \lor p q \equiv p (T \lor q) \equiv p$
% \item $p (p \lor q) \equiv p \lor p q \equiv p$ (from equation (\ref{eqn:b}))
% \item $p (U \lor \neg p q) \equiv U p \lor p \neg p q \equiv U p \lor U p \neg p q \equiv U p$ (from equation (\ref{eqn:b}))
%\end{enumerate}
%\end{proof}

\subsection{Allocation} \label{sec:allocation2}
We apply the allocation to the process of mapping each argument to an expression, and each mapping instance is called an \emph{allocator}.
The completeness of the allocation is defined in a manner similar to that of the labeling.

\begin{defn} \label{def:allocator}
An allocator $E$ for $\langle \AR, \ATT \rangle$ is complete iff for any argument $A \in \AR$, $E(A) \equiv \bigwedge_{(A', A) \in \ATT} \neg E(A')$.
\end{defn}
Note that if there are no $A'$ satisfying $(A', A) \in \ATT$, $E(A) \equiv T$.

For example, (a) in Figure~\ref{fig:AFexample1} has a complete allocator $E$ such that $E(1) = E(3) = E(5) = T$ and $E(2) = E(4) = F$.
There are, of course, an infinite number of allocators such that $E(3) \equiv T$ but $E(3) \neq T$.
For the purposes of this paper, however, equivalent allocators are irrelevant.

The following theorem demonstrates that allocation constitutes a generalization of labeling.
\begin{thrm} \label{thm:label}
For any complete labeling $L$, the allocator $E$ such that $E(A) = T$ iff $L(A) = \Lin$, $E(A) = F$ iff $L(A) = \Lout$, and $E(A) = U$ iff $L(A) = \Lundec$ is complete.
\end{thrm}

A complete allocator mapping only logical constants ($T$, $F$ or $U$) is called a \emph{constant} allocator.
The inverse of the above theorem is also valid.
\begin{thrm} \label{thm:labelinverse}
For any constant allocator $E$, labeling $L$ such that $L(A) = \Lin$ if $E(A) = T$, $L(A) = \Lout$ if $E(A) = F$, and $L(A) = \Lundec$ if $E(A) = U$ is a complete labeling.
\end{thrm}

\begin{example} \label{example:variable}
(b) in Figure~\ref{fig:AFexample1} has three constant allocators corresponding to complete labelings.
Simultaneously, it also has a complete allocator $E$ such that $E(1) = a$, $E(2) = \neg a$, $E(3) = a \neg a$ and $E(4) = a \lor \neg a$.
\begin{itemize}
 \item When the valuation $v$ is defined as $v(a) = T$, these expressions are evaluated to $T$, $F$, $F$, $T$, respectively. This result corresponds to $L_2$, as described in Example~\ref{example:labeling}.
 \item When $v(a) = F$, they are $F$, $T$, $F$, $T$, respectively. This result corresponds to $L_3$.
 \item When $v(a) = U$, they are all $U$. This result corresponds to $L_1$.
\end{itemize}
\end{example}
The purpose of this observation is to show that $E$ abstracts these labelings.
Note that a complete allocator is required to allocate logical expressions with variables to arguments \emph{only if} the arguments are in cycles of attack relations (see Theorem~\ref{thm:grounded}).

We demonstrate the relationship between $E$ and constant allocators corresponding to the labelings.
For any complete allocator, a variable may be replaced with an expression.
\begin{thrm} \label{thm:instantiate}
For the allocator $E$, we write $E[p/x](A) = E(A)[p/x]$ where $p$ is an expression and $x$ is a variable.
Consequently, for any complete allocator $E$, $E[p/x]$ is also complete.
\end{thrm}

Hereafter, we call the set of variables occurring in expressions allocated by allocator $E$, i.e. $\{ x | x \mbox{ occurs in } E(A) \mbox{ for some } A \in \AR \}$, the \emph{allocation variables} of $E$.
By replacing a variable with a logical constant, we obtain a complete allocator based on another complete allocator that has more allocation variables.
When a complete allocator $E'$ is equivalent to another allocator $E$, with some variables substituted by constants, $E'$ is said to be instantiated from $E$.
When all allocation variables of a complete allocator are replaced, the resulting allocator is equivalent to a constant allocator.

\begin{thrm} \label{thm:instantiatevaluation}
For any complete allocator $E$ and valuation $v$, an allocator $E_v$ which is defined as $E_v(A) = \evalV{E(A)}{v}$ is a constant allocator.
\end{thrm}

There are some arguments to which any complete allocator allocates expressions that are equivalent to logical constants.
For example, the AF $\langle \{1, 2, 3\},$ $\{(1, 2), (2, 3), (3, 1)\} \rangle$ has a unique complete allocator $E$, such that $E(A) \equiv U$ for $A = 1,2,3$.
The following theorem shows another example.

\begin{thrm} \label{thm:grounded}
Let $L_g$ be a grounded labeling of a framework and $E$ be a complete allocator.
If $L_g(A) = \Lin$, then $E(A) \equiv T$.
Also, if $L_g(A) = \Lout$, then $E(A) \equiv F$.
\end{thrm}

Based on this instantiation, we arrive at the notion of general allocators.
\begin{defn} \label{def:general}
A complete allocator $E$ is called a general allocator iff for any constant allocator $E'$ there exists a valuation $v$ such that $E' = E_v$\footnote{In \cite{Moriguchi18}, we use complete labelings rather than constant allocators, but, according to Theorems~\ref{thm:label} and \ref{thm:labelinverse}, these are the same.}.
\end{defn}

We have already encountered an example of a general allocator as $E$, described in Example~\ref{example:variable}.

In \cite{Moriguchi18}, we proved the existence of general allocators for any framework.
\begin{thrm}
There is a general allocator for each finite framework.
\end{thrm}
The proof is derived from the following theorem.

\begin{thrm} \label{thm:generalize}
Let $E_1$ and $E_2$ be complete allocators, $a$ be a fresh variable, and $L_g$ be the grounded labeling.
When the allocator $E$ is defined as $E(A) = a E_1(A) \lor \neg a E_2(A) \lor a \neg a$ if $L_g(A) = \Lundec$ and $E(A) = E_1(A)$ otherwise,
$E$ is complete.
\end{thrm}
%\begin{proof}
%For valuation $v$,
%\begin{itemize}
% \item Let $v(a) = 1$. If argument $A$ satisfies $L_g(A) = \Lundec$, then $\evalV{E(A)}{v} = \evalV{E_1(A)}{v}$. Otherwise, $E(A) = E_1(A)$, so $\evalV{E(A)}{v} = \evalV{E_1(A)}{v}$.
% Since $E_1$ is a complete allocator, $E_1(A) \equiv \bigwedge_{(A', A) \in \ATT} \neg E_1(A')$.
% So $\evalV{E(A)}{v}  = \evalV{E_1(A)}{v} = \evalV{\bigwedge_{(A', A) \in \ATT} \neg E_1(A')}{v} = \evalV{\bigwedge_{(A', A) \in \ATT} \neg E(A')}{v}$.
%
% \item Let $v(a) = -1$. If argument $A$ satisfies $L_g(A) = \Lundec$, then $\evalV{E(A)}{v} = \evalV{E_2(A)}{v}$. And otherwise $E(A) = E_1(A)$, and from theorem~\ref{thm:grounded} $E_1(A) = E_2(A)$, so $\evalV{E(A)}{v} = \evalV{E_2(A)}{v}$.
% Applying the same process, $\evalV{E(A)}{v} = \evalV{\bigwedge_{(A', A) \in \ATT} \neg E(A')}{v}$.
%
% \item Let $v(a) = 0$. If argument $A$ satisfies $L_g(A) = \Lundec$, then $\evalV{E(A)}{v} = 0$. Otherwise $E(A) = E_1(A)$, and from theorem~\ref{thm:grounded} $\evalV{E_1(A)}{v} = 1$ if $L_g(A) = \Lin$ and $\evalV{E_1(A)} = -1$ if $L_g(A) = \Lout$. This corresponds to $L_g$; according to theorem~\ref{thm:label} it is complete.
%\end{itemize}
%Therefore $E$ is complete.
%\end{proof}

A composed allocator in Theorem~\ref{thm:generalize} behaves like $E_1$ when $a$ is evaluated to $T$, $E_2$ when $a$ is evaluated to $F$, and ground labeling (corresponding constant allocator) when $a$ is evaluated to $U$.
So, the process for building a general allocator composes all constant allocators, except for the one corresponding to the ground labeling.

The framework in Example~\ref{example:labeling} has three constant allocators corresponding to its labelings.
$E_1(A) = U$, $E_2(1) = E_2(4) = T$ and $E_2(2) = E_2(3) = F$, $E_3(2) = E_3(4) = T$ and $E_3(1) = E_3(3) = F$.
$E_1$ corresponds to the ground labeling.
Here, we compose $E_2$ and $E_3$ as $E$.
All arguments are labeled $\Lundec$ by the ground labeling, $E$ is defined as follows.
\begin{itemize}
 \item $E(1) = a T \lor \neg a F \lor a \neg a \equiv a \lor a \neg a \equiv a$
 \item $E(2) = a F \lor \neg a T \lor a \neg a \equiv \neg a \lor a \neg a \equiv \neg a$
 \item $E(3) = a F \lor \neg a F \lor a \neg a \equiv a \neg a$
 \item $E(4) = a T \lor \neg a T \lor a \neg a \equiv a \lor \neg a$
\end{itemize}
From $E$ and $v_U$, $E_1$ is instantiated.
$E_2$ and $E_3$ are instantiated from the valuations $v_2$ and $v_3$, where $v_2(a) = T$ and $v_3(a) = F$.
Therefore, $E$ is a general allocator.

We now present another example, as shown in Figure~\ref{fig:allocation_example}.
There are five constant allocators (i.e., complete labelings).
\begin{itemize}
 \item $E_1(A) = U$ for $A \in \{ 1, 2, 3, 4, 5 \}$.
 \item $E_2(A) = U$ for $A \in \{ 1, 2, 3 \}$, $E_2(4) = F$ and $E_2(5) = T$.
 \item $E_3(1) = T$, $E_3(2) = E_3(3) = F$ and $E_3(A) = U$ for $A \in \{ 4, 5 \}$.
 \item $E_4(1) = T$, $E_4(2) = E_4(3) = F$, $E_4(4) = T$ and $E_4(5) = F$.
 \item $E_5(1) = T$, $E_5(2) = E_5(3) = F$, $E_5(4) = F$ and $E_5(5) = T$.
\end{itemize}

\begin{figure}[tbp]
 \begin{center}
 \includegraphics[height=2cm]{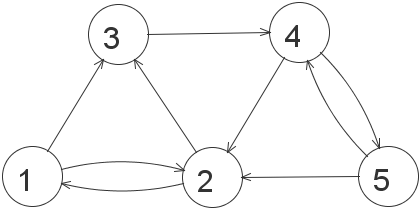}
 \end{center}
 \caption{The argumentation framework with five arguments and five complete labelings.}
 \label{fig:allocation_example}
\end{figure}

By composing $E_2$ and $E_3$, we obtain $E_{23}$ such that $E_{23}(1) = a_1 U \lor \neg a_1 T \lor a_1 \neg a_1 \equiv U a_1 \lor \neg a_1$\footnote{Here we show the allocated expression for $1$.}.
Next, composition with $E_4$ gives $E_{234}$ such that $E_{234}(1) = (U a_1 \lor \neg a_1) a_2 \lor \neg a_2 T \lor a_2 \neg a_2 \equiv (U a_1 \lor \neg a_1) a_2 \lor \neg a_2$.
Finally, we obtain a general allocator $E$ by composing $E_5$, which is $E(1) = ((U a_1 \lor \neg a_1) a_2 \lor \neg a_2) a_3 \lor \neg a_3$.

There are two problems with this method; the size of the allocator and the computational complexity.

Unfortunately, the general allocator produced by Theorem~\ref{thm:generalize} is frequently larger than desired.
As shown above, we need $n-2$ variables to build a general allocator, where $n$ is the number of complete labelings.
The AF shown in Figure~\ref{fig:AFexample3} has nine complete labelings, and thus we need seven variables.
However, as the left part (\{ 1, 2 \}) and right part (\{ 3, 4 \}) only affect the central argument $5$, without affecting each other, the general allocator $E$ can be constructed using only two variables, as $E(1) = \neg a$, $E(2) = a$, $E(3) = \neg b$, $E(4) = b$ and $E(5) = \neg a \lor \neg b$.

\begin{figure}[tbp]
 \begin{center}
 \includegraphics[height=2cm]{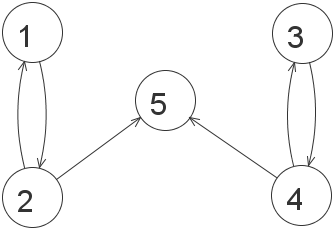}
 \end{center}
 \caption{A framework with two cyclic components.}
 \label{fig:AFexample3}
\end{figure}

Another problem with this process is that its complexity depends on the enumeration of complete labelings, which is not processed in polynomial time (unless $P \neq \mathit{NP}$)~\cite{Kroll17}.
To implement the allocation method, we require an algorithm for constructing general allocators without enumerating complete labelings.

\section{Equation Solving} \label{sec:equation}
In this section, we propose a new method for constructing general allocators for the framework.
Based on the definition of completeness, we can consider constant allocators and general allocators as particular solutions and general solutions of simultaneous equations, $E(A) \equiv \bigwedge_{(A', A) \in \ATT} \neg E(A')$, respectively.
We solve the equations by transformations and substitutions and obtain general allocators, as for numerical equations.

\subsection{Equations}
Here, we focus on the specific form of the equations, as the following definition\footnote{As we use $=$ for both an identity and an equation in numerical systems, we use $\equiv$ for both an equivalence relation and an equation.}.
\begin{defn}[Variable Equation]
An equation $X \equiv G$ where $X$ is a variable\footnote{In this section, we use $X$ to denote a variable on the left side of an equation.} and $G$ is an expression, is called a variable equation.
A solution of the equation is a valuation $v$ for $\Var \supseteq \{ X \} \cup \Var(G)$ where $v(X) = \evalV{G}{v}$.
We use the functions $\mathrm{LV}$ and $\mathrm{RE}$ over variable equations, $\LV{X \equiv G} = X$ and $\RE{X \equiv G} = G$.
\end{defn}

The following corresponds to the notion of simultaneous equations.

\begin{defn}[Variable Equation Set]
A set of variable equations $\mathbf{E}$ where $\forall E_1, E_2 \in \mathbf{E}. E_1 \neq E_2 \Rightarrow \LV{E_1} \neq \LV{E_2}$ is called a variable equation set.
A solution of the variable equation set is a valuation $v$ for $\Var$ where $v$ is a solution of any equation in $\mathbf{E}$.
We use the functions $\mathrm{LV}$ and $\mathrm{RE}$ as for variable equations, $\LV{\mathbf{E}} = \{ \LV{E} | E \in \mathbf{E} \}$ and $\RE{\mathbf{E}} = \{ \RE{E} | E \in \mathbf{E} \}$.
We also use the notation $\mathbf{E}(X)$ as $E$ where $E \in \mathbf{E}$ and $\LV{E} = X$.
\end{defn}

A variable equation set does not allow two variable equations about the same variable.
This limitation is sufficiently general to define the conditions for complete allocators. %%% 最初はdenoteと書いていたがdefineに直された。describeの方が良さそうな気もするが・・・

\begin{lmm} \label{lmm:solution_constant}
For the framework $\langle \AR, \ATT \rangle$, $\{ A \equiv \bigwedge_{(A', A) \in \ATT} \neg A' | A \in \AR \}$ can be seen as a variable equation set where $\Var = \AR$.
Then, a solution of the variable equation set is a constant allocator for the framework, and vice versa.
\end{lmm}
This conclusion follows straightforwardly from the definitions of constant allocators and solutions of variable equation sets.

\subsection{Transformation}
Variables in equations depend on themselves through other variables and equations.
We call these types of dependencies in the variables \emph{self-dependencies}.
In numerical equations, we add/multiply some subexpressions to both sides of the equations to cancel out variables on one side and substitute the equivalent expressions to variables.
Because we do not have cancelable operators in three-valued logic, it is difficult to solve such equations in this way.
Here, we propose a similar equation for the given variable equation.

In this section, we make frequent use of variables in equation sets.
We use the function symbol $\Var$ over expressions, variable equations and variable equation sets, to denote variables occurring in each of them.

\begin{defn}[Refined Equation] \label{defn:refined_equation}
Let $X$ be a variable, $G$ and $H$ be expressions satisfying $\Var(G) \subseteq \Var(H) \cup \{ X \}$ and $X \not\in \Var(H)$, and $V$ be a set of all valuations for $\Var \supseteq \Var(H) \cup \{ X \}$.
$X \equiv H$ is called a refined equation of $X \equiv G$ when the following two formulas are satisfied.
\begin{itemize}
 \item $\forall v \in V. v(X) = \evalV{H}{v} \Rightarrow v(X) = \evalV{G}{v}$
 \item $\forall v \in V. v(X) = \evalV{G}{v} \Rightarrow \exists v' \in V. v(X) = \evalV{H}{v} \land \forall x \in (\Var \setminus \Var(H)) \cup \Var(G). v(x) = v'(x)$
\end{itemize}
\end{defn}
When $X \not\in \Var(G)$, $X \equiv G$ is a refined equation of itself.

The conditions in the above definition state that, in the former case, any solution of the refined equation is also the solution of the original equation; and in the latter case, that any solution of the original equation is represented by at least one solution of the refined equation.
In other words, the solutions of the refined equation cover all solutions of the original equation without self-dependency of $X$.

Self-dependency of the variable on the left side allows \emph{some} assignments in the solutions of an equation, without changing the assignments of other variables.
For example, $X \equiv U X$ gives two solutions, $v_1(X) = T$ and $v_2(X) = U$.
Because a refined equation removes self-dependency from the original equation, we need an extra variable to emulate such behavior.
In this case, $X \equiv U x$ is a refined equation where $x$ is an extra variable.

To obtain a refined equation, we need to clarify the self-dependency in the equation.
So, we build a refined equation of an arbitrary equation through the intermediate form, $X \equiv P X \lor N \neg X \lor C X \neg X \lor M$\footnote{$P$, $N$, and $C$ are coefficients of $X$ (positive dependency), $\neg X$ (negative dependency), and $X \neg X$ (complex dependency), respectively. $C$ is required because, as shown in Lemma~\ref{lmm:three-valued-not-equivalence}, $X \neg X$ is not equivalent to $F$.}.
First, we introduce the function used to transform an expression into this form.

\begin{defn} \label{defn:transform_function}
Assume that $x$ and $y$ are variables where $x \neq y$.
$\mathcal{R}_x$ is a function from an expression to a quadruple of expressions.
$$\begin{array}{lrl@{\quad}lrl}
\mathcal{R}_x(c)                        & = & (F, F, F, c)                & \mathcal{R}_x(p_1 \lor p_2)             & = & (P_1 \lor P_2, N_1 \lor N_2, C_1 \lor C_2, M_1 \lor M_2) \\
\mathcal{R}_x(x)                        & = & (T, F, F, F)                & \mathcal{R}_x(p_1 p_2)                  & = & (P_1 P_2 \lor P_1 M_2 \lor P_2 M_1,   \\
\mathcal{R}_x(y)                        & = & (F, F, F, y)                &                                         &   & \ \  N_1 N_2 \lor N_1 M_2 \lor N_2 M_1, \\
\mathcal{R}_x(\neg p)                   & = & \overline{\mathcal{R}}_x(p) &                                         &   & \ \  (P_1 \lor N_1 \lor M_1) C_2 \lor (P_2 \lor N_2 \lor M_2) C_1  \\
                                        &   &                             &                                         &   & \qquad \lor C_1 C_2 \lor P_1 N_2 \lor N_1 P_2, \\
                                        &   &                             &                                         &   & \ \  M_1 M_2) \\
\overline{\mathcal{R}}_x(c)             & = & (F, F, F, \neg c)           & \overline{\mathcal{R}}_x(p_1 p_2) & = & (\overline{P}_1 \lor \overline{P}_2, \overline{N}_1 \lor \overline{N}_2, \overline{C}_1 \lor \overline{C}_2, \overline{M}_1 \lor \overline{M}_2) \\
\overline{\mathcal{R}}_x(x)             & = & (F, T, F, F)                & \overline{\mathcal{R}}_x(p_1 \lor p_2)  & = & (\overline{P}_1 \overline{P}_2 \lor \overline{P}_1 \overline{M}_2 \lor \overline{P}_2 \overline{M}_1, \\
\overline{\mathcal{R}}_x(y)             & = & (F, F, F, \neg y)           &                                         &   & \ \  \overline{N}_1 \overline{N}_2 \lor \overline{N}_1 \overline{M}_2 \lor \overline{N}_2 \overline{M}_1, \\
\overline{\mathcal{R}}_x(\neg p)        & = & \mathcal{R}_x(p)            &                                         &   & \ \  (\overline{P}_1 \lor \overline{N}_1 \lor \overline{M}_1) \overline{C}_2 \lor (\overline{P}_2 \lor \overline{N}_2 \lor \overline{M}_2) \overline{C}_1 \\
                                        &   &                             &                                         &   & \qquad \lor \overline{C}_1 \overline{C}_2 \lor \overline{P}_1 \overline{N}_2 \lor \overline{N}_1 \overline{P}_2, \\
                                        &   &                             &                                         &   & \ \  \overline{M}_1 \overline{M}_2) \\
\end{array}$$
where $c$ is a logical constant, $(P_i, N_i, C_i, M_i) = \mathcal{R}_x(p_i)$ and $(\overline{P}_i, \overline{N}_i, \overline{C}_i, \overline{M}_i) = \overline{\mathcal{R}}_x(p_i)$ for $i = 1, 2$.
\end{defn}
Intuitively, $\mathcal{R}$ processes an expression $G$ and $\overline{\mathcal{R}}$ processes an expression $\neg G$.
The following lemma demonstrates the intention of these functions.

\begin{lmm} \label{lmm:PNCMform}
Let $X$ be a variable and $G$ be an expression.
When $\mathcal{R}_X(G) = (P, N, C, M)$, the following conditions are valid.
\begin{itemize}
 \item There are no occurrences of $X$ in $P$, $N$, $C$ or $M$.
 \item $G \equiv P X \lor N \neg X \lor C X \neg X \lor M$.
 \item $\Var(P) \cup \Var(N) \cup \Var(C) \cup \Var(M) = \Var(G) \setminus \{ X \}$.
\end{itemize}
\end{lmm}
This is proven by induction on $G$\footnote{Precisely speaking, we should declare similar propositions about $\overline{\mathcal{R}}$.}.

Here, we show some small examples of refined equations.
As shown above, $X \equiv U x$ is a refined equation of $X \equiv U X$.
In general, $X \equiv P x$ (where $P$ is an expression without any occurrences of $X$ and $x$) is a refined equation of $X \equiv P X$, because when $P$ is evaluated to $T$, $U$, and $F$, $X$ can be one of $\{ T, U, F \}$, $\{ U, F \}$, and $\{ F \}$, respectively, and $P x$ is also in the same range in each case.

$X \equiv \neg X$ only gives one solution, $v(X) = U$, so $X \equiv U$ is a refined equation of that equation.
The case of $X \equiv N \neg X$, where $N$ is an expression without any occurrences of $X$, yields two types of solution; one is $v(X) = U$ and $\evalV{N}{v'} \neq F$, and the other is $v'(X) = F$ and $\evalV{N}{v'} = F$.
In this case, $X \equiv U N$ is a refined equation of $X \equiv N \neg X$.

Fortunately, the case $X \equiv P X \lor N \neg X$ is easy; $X \equiv P x \lor U N$ is a refined equation of the equation.
Also, in a similar manner to the above, we can prove that $X \equiv U C x$ is a refined equation of $X \equiv C X \neg X$.

From the above, we obtain a refined equation of the equation of this form.

\begin{thrm} \label{thrm:PNCMequation}
Let $X$ be a variable and $P$, $N$, $C$ and $M$ be expressions that do not contain $X$.
When $x$ is a fresh variable (i.e., no occurrences in any of $P$, $N$, $C$, or $M$),  $X \equiv P x \lor U (N \lor C x) \lor M$ is a refined equation of $X \equiv P X \lor N \neg X \lor C X \neg X \lor M$.
\end{thrm}
\begin{proof}
For any valuation $v$, $P$, $N$, $C$ and $M$ will be evaluated to one of $T$, $F$ or $U$.
To verify that the conditions in Definition~\ref{defn:refined_equation} are satisfied, we should show that, for each evaluation of $P$, $N$, $C$ and $M$, there exist some assignments of $x$ for each possible assignment of $X$ that make the original equation and the refined equation valid (for the latter condition), and there exists a valid assignment of $X$ for each assignment of $x$ (for the former condition).

Table~\ref{tbl:PNCM_tbl} shows the evaluation results of $P$, $N$, $C$ and $M$, the possible assignments for $X$ that make the original equation valid, and the assignments of $x$ that make the refined equation valid.
Such a valuation only differs from the original valuation in terms of the assignment of $x$.
You can see that there exist some assignments of $x$ for each column (i.e., the latter condition in the definition of refined equations is satisfied), and each assignment of $x$ occurs exactly once in the result of each evaluation of $P$, $N$, $C$, and $M$ (i.e., the former condition is satisfied).
\end{proof}

\begin{table}[ptb]
\caption{Correspondence between valuations for two equations in Theorem~\ref{thrm:PNCMequation}.
The columns of $P$, $N$, $C$ and $M$ denote the evaluation results of each term with valuations, that of possible-$X$ denotes the valid assignments of the valuations, and that of $x$ denotes the assignments of $x$ that makes the refined equation valid.} %%% that of $x$ denotes the assignment"s"にしたのは意図的で英文校正は単数にしていたが、実のところ複数あり得るのでsつく方が多分正しい(この場合denoteが正しいかは不明)。表は校正から消してるので仕方ない。
\label{tbl:PNCM_tbl}
{ \tiny
\begin{tabular}[t]{llllll} \hline
$P$                  & $N$                  & $C$                  & $M$                  & possible-$X$ & $x$ \\ \hline
$T$                  & $T$                  & $T$                  & $T$                  & $T$          & $T$, $U$, $F$ \\ \hline
\multirow{2}{*}{$T$} & \multirow{2}{*}{$T$} & \multirow{2}{*}{$T$} & \multirow{2}{*}{$U$} & $T$          & $T$ \\
                     &                      &                      &                      & $U$          & $U$, $F$ \\ \hline
\multirow{2}{*}{$T$} & \multirow{2}{*}{$T$} & \multirow{2}{*}{$T$} & \multirow{2}{*}{$F$} & $T$          & $T$ \\
                     &                      &                      &                      & $U$          & $U$, $F$ \\ \hline
$T$                  & $T$                  & $U$                  & $T$                  & $T$          & $T$, $U$, $F$ \\ \hline
\multirow{2}{*}{$T$} & \multirow{2}{*}{$T$} & \multirow{2}{*}{$U$} & \multirow{2}{*}{$U$} & $T$          & $T$ \\
                     &                      &                      &                      & $U$          & $U$, $F$ \\ \hline
\multirow{2}{*}{$T$} & \multirow{2}{*}{$T$} & \multirow{2}{*}{$U$} & \multirow{2}{*}{$F$} & $T$          & $T$ \\
                     &                      &                      &                      & $U$          & $U$, $F$ \\ \hline
$T$                  & $T$                  & $F$                  & $T$                  & $T$          & $T$, $U$, $F$ \\ \hline
\multirow{2}{*}{$T$} & \multirow{2}{*}{$T$} & \multirow{2}{*}{$F$} & \multirow{2}{*}{$U$} & $T$          & $T$ \\
                     &                      &                      &                      & $U$          & $U$, $F$ \\ \hline
\multirow{2}{*}{$T$} & \multirow{2}{*}{$T$} & \multirow{2}{*}{$F$} & \multirow{2}{*}{$F$} & $T$          & $T$ \\
                     &                      &                      &                      & $U$          & $U$, $F$ \\ \hline
$T$                  & $U$                  & $T$                  & $T$                  & $T$          & $T$, $U$, $F$ \\ \hline
\multirow{2}{*}{$T$} & \multirow{2}{*}{$U$} & \multirow{2}{*}{$T$} & \multirow{2}{*}{$U$} & $T$          & $T$ \\
                     &                      &                      &                      & $U$          & $U$, $F$ \\ \hline
\multirow{2}{*}{$T$} & \multirow{2}{*}{$U$} & \multirow{2}{*}{$T$} & \multirow{2}{*}{$F$} & $T$          & $T$ \\
                     &                      &                      &                      & $U$          & $U$, $F$ \\ \hline
$T$                  & $U$                  & $U$                  & $T$                  & $T$          & $T$, $U$, $F$ \\ \hline
\multirow{2}{*}{$T$} & \multirow{2}{*}{$U$} & \multirow{2}{*}{$U$} & \multirow{2}{*}{$U$} & $T$          & $T$ \\
                     &                      &                      &                      & $U$          & $U$, $F$ \\ \hline
\multirow{2}{*}{$T$} & \multirow{2}{*}{$U$} & \multirow{2}{*}{$U$} & \multirow{2}{*}{$F$} & $T$          & $T$ \\
                     &                      &                      &                      & $U$          & $U$, $F$ \\ \hline
$T$                  & $U$                  & $F$                  & $T$                  & $T$          & $T$, $U$, $F$ \\ \hline
\multirow{2}{*}{$T$} & \multirow{2}{*}{$U$} & \multirow{2}{*}{$F$} & \multirow{2}{*}{$U$} & $T$          & $T$ \\
                     &                      &                      &                      & $U$          & $U$, $F$ \\ \hline
\multirow{2}{*}{$T$} & \multirow{2}{*}{$U$} & \multirow{2}{*}{$F$} & \multirow{2}{*}{$F$} & $T$          & $T$ \\
                     &                      &                      &                      & $U$          & $U$, $F$ \\ \hline
$T$                  & $F$                  & $T$                  & $T$                  & $T$          & $T$, $U$, $F$ \\ \hline
\multirow{2}{*}{$T$} & \multirow{2}{*}{$F$} & \multirow{2}{*}{$T$} & \multirow{2}{*}{$U$} & $T$          & $T$ \\
                     &                      &                      &                      & $U$          & $U$, $F$ \\ \hline
\multirow{3}{*}{$T$} & \multirow{3}{*}{$F$} & \multirow{3}{*}{$T$} & \multirow{3}{*}{$F$} & $T$          & $T$ \\
                     &                      &                      &                      & $U$          & $U$ \\
                     &                      &                      &                      & $F$          & $F$ \\ \hline
$T$                  & $F$                  & $U$                  & $T$                  & $T$          & $T$, $U$, $F$ \\ \hline
\multirow{2}{*}{$T$} & \multirow{2}{*}{$F$} & \multirow{2}{*}{$U$} & \multirow{2}{*}{$U$} & $T$          & $T$ \\
                     &                      &                      &                      & $U$          & $U$, $F$ \\ \hline
\multirow{3}{*}{$T$} & \multirow{3}{*}{$F$} & \multirow{3}{*}{$U$} & \multirow{3}{*}{$F$} & $T$          & $T$ \\
                     &                      &                      &                      & $U$          & $U$ \\
                     &                      &                      &                      & $F$          & $F$ \\ \hline
$T$                  & $F$                  & $F$                  & $T$                  & $T$          & $T$, $U$, $F$ \\ \hline
\multirow{2}{*}{$T$} & \multirow{2}{*}{$F$} & \multirow{2}{*}{$F$} & \multirow{2}{*}{$U$} & $T$          & $T$ \\
                     &                      &                      &                      & $U$          & $U$, $F$ \\ \hline
\multirow{3}{*}{$T$} & \multirow{3}{*}{$F$} & \multirow{3}{*}{$F$} & \multirow{3}{*}{$F$} & $T$          & $T$ \\
                     &                      &                      &                      & $U$          & $U$ \\
                     &                      &                      &                      & $F$          & $F$ \\ \hline
$U$                  & $T$                  & $T$                  & $T$                  & $T$          & $T$, $U$, $F$ \\ \hline
$U$                  & $T$                  & $T$                  & $U$                  & $U$          & $T$, $U$, $F$ \\ \hline
$U$                  & $T$                  & $T$                  & $F$                  & $U$          & $T$, $U$, $F$ \\ \hline
$U$                  & $T$                  & $U$                  & $T$                  & $T$          & $T$, $U$, $F$ \\ \hline
$U$                  & $T$                  & $U$                  & $U$                  & $U$          & $T$, $U$, $F$ \\ \hline
$U$                  & $T$                  & $U$                  & $F$                  & $U$          & $T$, $U$, $F$ \\ \hline
\end{tabular}
\quad
\begin{tabular}[t]{llllll} \hline
$P$                  & $N$                  & $C$                  & $M$                  & possible-$X$ & $x$ \\ \hline
$U$                  & $T$                  & $F$                  & $T$                  & $T$          & $T$, $U$, $F$ \\ \hline
$U$                  & $T$                  & $F$                  & $U$                  & $U$          & $T$, $U$, $F$ \\ \hline
$U$                  & $T$                  & $F$                  & $F$                  & $U$          & $T$, $U$, $F$ \\ \hline
$U$                  & $U$                  & $T$                  & $T$                  & $T$          & $T$, $U$, $F$ \\ \hline
$U$                  & $U$                  & $T$                  & $U$                  & $U$          & $T$, $U$, $F$ \\ \hline
$U$                  & $U$                  & $T$                  & $F$                  & $U$          & $T$, $U$, $F$ \\ \hline
$U$                  & $U$                  & $U$                  & $T$                  & $T$          & $T$, $U$, $F$ \\ \hline
$U$                  & $U$                  & $U$                  & $U$                  & $U$          & $T$, $U$, $F$ \\ \hline
$U$                  & $U$                  & $U$                  & $F$                  & $U$          & $T$, $U$, $F$ \\ \hline
$U$                  & $U$                  & $F$                  & $T$                  & $T$          & $T$, $U$, $F$ \\ \hline
$U$                  & $U$                  & $F$                  & $U$                  & $U$          & $T$, $U$, $F$ \\ \hline
$U$                  & $U$                  & $F$                  & $F$                  & $U$          & $T$, $U$, $F$ \\ \hline
$U$                  & $F$                  & $T$                  & $T$                  & $T$          & $T$, $U$, $F$ \\ \hline
$U$                  & $F$                  & $T$                  & $U$                  & $U$          & $T$, $U$, $F$ \\ \hline
\multirow{2}{*}{$U$} & \multirow{2}{*}{$F$} & \multirow{2}{*}{$T$} & \multirow{2}{*}{$F$} & $U$          & $T$, $U$ \\
                     &                      &                      &                      & $F$          & $F$ \\ \hline
$U$                  & $F$                  & $U$                  & $T$                  & $T$          & $T$, $U$, $F$ \\ \hline
$U$                  & $F$                  & $U$                  & $U$                  & $U$          & $T$, $U$, $F$ \\ \hline
\multirow{2}{*}{$U$} & \multirow{2}{*}{$F$} & \multirow{2}{*}{$U$} & \multirow{2}{*}{$F$} & $U$          & $T$, $U$ \\
                     &                      &                      &                      & $F$          & $F$ \\ \hline
$U$                  & $F$                  & $F$                  & $T$                  & $T$          & $T$, $U$, $F$ \\ \hline
$U$                  & $F$                  & $F$                  & $U$                  & $U$          & $T$, $U$, $F$ \\ \hline
\multirow{2}{*}{$U$} & \multirow{2}{*}{$F$} & \multirow{2}{*}{$F$} & \multirow{2}{*}{$F$} & $U$          & $T$, $U$ \\
                     &                      &                      &                      & $F$          & $F$ \\ \hline
$F$                  & $T$                  & $T$                  & $T$                  & $T$          & $T$, $U$, $F$ \\ \hline
$F$                  & $T$                  & $T$                  & $U$                  & $U$          & $T$, $U$, $F$ \\ \hline
$F$                  & $T$                  & $T$                  & $F$                  & $U$          & $T$, $U$, $F$ \\ \hline
$F$                  & $T$                  & $U$                  & $T$                  & $T$          & $T$, $U$, $F$ \\ \hline
$F$                  & $T$                  & $U$                  & $U$                  & $U$          & $T$, $U$, $F$ \\ \hline
$F$                  & $T$                  & $U$                  & $F$                  & $U$          & $T$, $U$, $F$ \\ \hline
$F$                  & $T$                  & $F$                  & $T$                  & $T$          & $T$, $U$, $F$ \\ \hline
$F$                  & $T$                  & $F$                  & $U$                  & $U$          & $T$, $U$, $F$ \\ \hline
$F$                  & $T$                  & $F$                  & $F$                  & $U$          & $T$, $U$, $F$ \\ \hline
$F$                  & $U$                  & $T$                  & $T$                  & $T$          & $T$, $U$, $F$ \\ \hline
$F$                  & $U$                  & $T$                  & $U$                  & $U$          & $T$, $U$, $F$ \\ \hline
$F$                  & $U$                  & $T$                  & $F$                  & $U$          & $T$, $U$, $F$ \\ \hline
$F$                  & $U$                  & $U$                  & $T$                  & $T$          & $T$, $U$, $F$ \\ \hline
$F$                  & $U$                  & $U$                  & $U$                  & $U$          & $T$, $U$, $F$ \\ \hline
$F$                  & $U$                  & $U$                  & $F$                  & $U$          & $T$, $U$, $F$ \\ \hline
$F$                  & $U$                  & $F$                  & $T$                  & $T$          & $T$, $U$, $F$ \\ \hline
$F$                  & $U$                  & $F$                  & $U$                  & $U$          & $T$, $U$, $F$ \\ \hline
$F$                  & $U$                  & $F$                  & $F$                  & $U$          & $T$, $U$, $F$ \\ \hline
$F$                  & $F$                  & $T$                  & $T$                  & $T$          & $T$, $U$, $F$ \\ \hline
$F$                  & $F$                  & $T$                  & $U$                  & $U$          & $T$, $U$, $F$ \\ \hline
\multirow{2}{*}{$F$} & \multirow{2}{*}{$F$} & \multirow{2}{*}{$T$} & \multirow{2}{*}{$F$} & $U$          & $T$, $U$ \\
                     &                      &                      &                      & $F$          & $F$ \\ \hline
$F$                  & $F$                  & $U$                  & $T$                  & $T$          & $T$, $U$, $F$ \\ \hline
$F$                  & $F$                  & $U$                  & $U$                  & $U$          & $T$, $U$, $F$ \\ \hline
\multirow{2}{*}{$F$} & \multirow{2}{*}{$F$} & \multirow{2}{*}{$U$} & \multirow{2}{*}{$F$} & $U$          & $T$, $U$ \\
                     &                      &                      &                      & $F$          & $F$ \\ \hline
$F$                  & $F$                  & $F$                  & $T$                  & $T$          & $T$, $U$, $F$ \\ \hline
$F$                  & $F$                  & $F$                  & $U$                  & $U$          & $T$, $U$, $F$ \\ \hline
$F$                  & $F$                  & $F$                  & $F$                  & $F$          & $T$, $U$, $F$ \\ \hline
\end{tabular}
}
\end{table}

From Lemma~\ref{lmm:PNCMform} and Theorem~\ref{thrm:PNCMequation}, we obtain a refined equation of a variable equation $E$.
We write the refined equation as $\REF{E}$, i.e., $\REF{E} = \LV{E} \equiv P x \lor U (N \lor C x) \lor M$, where $(P, N, C, M) = \mathcal{R}_{\LV{E}}(\RE{E})$.

\begin{coro} \label{coro:ref_variable}
Let $E$ be a variable equation and $x$ be a variable introduced in $\REF{E}$.
Then, $\Var(\RE{E}) \setminus \{ \LV{E} \} = \Var(\RE{\REF{E}}) \setminus \{ x \}$.
\end{coro}

\subsection{Substitution}
As the refined equation removes the self-dependencies, we can replace all occurrences of variables in $\RE{\mathbf{E}}$ with their expressions in the refined equation.

\begin{defn}[Substitution on Equations]
Let $\mathbf{E}$ be a variable equation set, $X$ be a variable, and $E$ be a refined equation of $\mathbf{E}(X)$ (i.e., $X = \LV{E}$).
We define substitution of equations as $\mathbf{E} [E] = \{ E \} \cup \{ \LV{E'} \equiv \RE{E'} [\RE{E}/X] | E' \in \mathbf{E} \setminus \{ \mathbf{E}(X) \} \}$.
Specifically, we write $\mathbf{E} [\REF{\mathbf{E}(X)}]$ as $\mathbf{E} [X]$.
\end{defn}

Substitution of the refined equation maintains the solutions, as a refined equation does.

\begin{thrm} \label{thrm:valuation_substitution}
Let $\mathbf{E}$ be a variable equation set, $X$ be a variable, $E$ be a refined equation of $\mathbf{E}(X)$ (i.e., $\LV{E} = X$), and $V$ be a set of all valuations for $\Var(E) \cup \Var(\mathbf{E})$.
If $(\Var(E) \setminus \Var(\mathbf{E}(X))) \cap \Var(\mathbf{E}) = \emptyset$, the following two formulas are satisfied.
\begin{itemize}
 \item $\forall v \in V. \  v \mbox{ is a solution of } \mathbf{E} [E] \Rightarrow \  v \mbox{ is a solution of } \mathbf{E}$
 \item $\forall v \in V. \  v \mbox{ is a solution of } \mathbf{E} \Rightarrow \exists v' \in V. \  v' \mbox{ is a solution of } \mathbf{E} [E] \land \forall x \in \mathrm{Var}(\mathbf{E}). v(x) = v'(x)$
\end{itemize}
\end{thrm}
\begin{proof}
For any valuation $v$ and variable equation $E' \in \mathbf{E} [E]$, $\evalV{\RE{E'} [\RE{E}/X]}{v} = \evalV{\RE{E'}}{v[\RE{E}/X]}$ (from Lemma~\ref{lmm:substitution}). %%% Overfull
If $v$ is a solution of $E$, then $\evalV{\RE{E}}{v} = v(X)$, i.e., $v[\RE{E}/X] = v$, and $\evalV{\RE{E'} [\RE{E}/X]}{v} = \evalV{\RE{E'}}{v}$.
As $E$ is included in both $\mathbf{E} [E]$ and $\{ E \} \cup \mathbf{E} \setminus \{ \mathbf{E}(X) \}$, if $v$ is a solution of one, then it is also a solution of the other.

When $v$ is a solution of $\{ E \} \cup \mathbf{E} \setminus \{ \mathbf{E}(X) \}$, it is also a solution of $\mathbf{E}(X)$ because $E$ is a refined equation of $\mathbf{E}(X)$.
This means that $v$ is a solution of $\mathbf{E}$.
When $v$ is a solution of $\mathbf{E}$, there exists a solution $v'$ for $E$ whose assignments for $(\Var \setminus \Var(E)) \cup \Var(\mathbf{E}(X))$ are the same as $v$.
The assumption $(\Var(E) \setminus \Var(\mathbf{E}(X))) \cap \Var(\mathbf{E}) = \emptyset$ shows that $v'$ is also a solution for $\mathbf{E} \setminus \{ \mathbf{E}(X) \}$.
This means that $v'$ is a solution of $\{ E \} \cup \mathbf{E} \setminus \{ \mathbf{E}(X) \}$.

Therefore, the two formulas are satisfied.
\end{proof}

\begin{lmm} \label{lmm:refine_decrease}
Let $\mathbf{E}$ be a variable equation set, $X \in \LV{\mathbf{E}}$ be a variable.
Then, $\LV{\mathbf{E} [X]} \cap \RE{\mathbf{E} [X]} = \LV{\mathbf{E}} \cap \RE{\mathbf{E}} \setminus \{ X \}$.
\end{lmm}
This is immediately proven from Corollary~\ref{coro:ref_variable}.

\subsection{General Solutions}

\begin{thrm} \label{thrm:valuation_general}
Let $\mathbf{E}$ be a variable equation set where $\LV{\mathbf{E}} \cap \Var(\RE{\mathbf{E}}) = \emptyset$.
The following function $\mathcal{A}$ is a bijection from valuations for $\Var(\RE{\mathbf{E}})$ to solutions of $\mathbf{E}$.
$$\mathcal{A}(v) = v' \mbox{ where } v'(x) = \left\{
 \begin{array}{l@{\quad}l}
  v(x) & (x \in \Var(\RE{\mathbf{E}})) \\
  \evalV{\mathbf{E}(x)}{v} & (x \in \LV{\mathbf{E}} \\
 \end{array}
\right .$$
\end{thrm}
\begin{proof}
For any two valuations $v$ and $v'$ for $\Var(\RE{\mathbf{E}})$, if $v \neq v'$ then there exists $x$ such that $v(x) \neq v'(x)$.
From the definition, $\mathcal{A}(v)(x) = v(x) \neq v'(x) = \mathcal{A}(v')(x)$, i.e, $\mathcal{A}(v) \neq \mathcal{A}(v')$ and therefore $\mathcal{A}$ is injective.

Assume that $v$ is a solution of $\mathbf{E}$.
Let $v'$ be the valuation of $\Var(\RE{\mathbf{E}})$ such that $v(x) = v'(x)$ for any $x \in \Var(\RE{\mathbf{E}})$.
For any variable $X \in \Var(\LV{\mathbf{E}})$,
$$\begin{array}{lrl@{\qquad}l}
\mathcal{A}(v')(X) & = & \evalV{\RE{\mathbf{E}(X)}}{v'} & \\
                   & = & \evalV{\RE{\mathbf{E}(X)}}{v} & (\mbox{from the definition of } v') \\
                   & = & v(X) & (v \mbox{ is a solution of } \mathbf{E}(X))
\end{array}$$
Hence, $\mathcal{A}(v') = v$ and therefore $\mathcal{A}$ is surjective.

Consequently, $\mathcal{A}$ is a bijection.
\end{proof}

According to Theorem~\ref{thrm:valuation_general}, we can focus on the valuations of the right-side expressions of the variable equation set, rather than its solutions (valuations for \emph{all} variables).
Enumerating the valuations shows all possible solutions, which can be considered constant allocators.
Therefore, the variable equation set that satisfies the conditions of Theorem~\ref{thrm:valuation_general} can be seen as a general allocator for the framework whose arguments are the left-side variables of the set.
Formally, the proposition is as follows.

\begin{thrm} \label{thrm:general_construction}
Let $\langle \AR, \ATT \rangle$ be a framework where $\AR = \{ A_1, A_2, \ldots A_n \}$.
$E(A) = \RE{\mathbf{E}(A)}$ where $\mathbf{E} = \{ A_i \equiv \bigwedge_{(A, A_i) \in \ATT} \neg A | 1 \leq i \leq n \} [A_1] [A_2] \ldots [A_n]$, is a general allocator of the framework.
\end{thrm}
\begin{proof}
For $1 \leq t \leq n$, we write $\mathbf{E}_t = \{ A_i \equiv \bigwedge_{(A, A_i) \in \ATT} \neg A | 1 \leq i \leq n \} [A_1] [A_2] \ldots [A_t]$.
From Lemma~\ref{lmm:refine_decrease},
$$\begin{array}{lrl}
\LV{\mathbf{E}} \cap \Var(\RE{\mathbf{E}}) & = & \LV{\mathbf{E}_n} \cap \Var(\RE{\mathbf{E}_n}) \\
                                           & = & \LV{\mathbf{E}_{n-1}} \cap \Var(\RE{\mathbf{E}_{n-1}}) \setminus \{ A_n \} \\
                                           & = & \LV{\mathbf{E}_{n-2}} \cap \Var(\RE{\mathbf{E}_{n-2}}) \setminus \{ A_{n-1}, A_n \} \\
                                           & = & \ldots \\
                                           & = & \LV{\mathbf{E}_1} \cap \Var(\RE{\mathbf{E}_1}) \setminus \{ A_2, \ldots, A_{n-1}, A_n \} \\
                                           & = & \AR \cap \Var(\{ \bigwedge_{(A, A_i) \in \ATT} \neg A | 1 \leq i \leq n \}) \setminus \AR \\
                                           & = & \emptyset
\end{array}$$
From Theorem~\ref{thrm:valuation_substitution}, each solution of $\mathbf{E}$ is denoted by the corresponding valuation of $\Var(\RE{\mathbf{E}})$.
By Theorem~\ref{thrm:valuation_general}, the valuations of $\Var(\RE{\mathbf{E}})$ give all of the solutions of the original set, $\{ A_i \equiv \bigwedge_{(A, A_i) \in \ATT} \neg A | 1 \leq i \leq n \}$.
The solutions of $\{ A_i \equiv \bigwedge_{(A, A_i) \in \ATT} \neg A | 1 \leq i \leq n \}$ correspond to the constant allocators for the framework (Lemma~\ref{lmm:solution_constant}), and therefore, $E(A) = \RE{\mathbf{E}(A)}$ is a general allocator.
\end{proof}

We use Figure~\ref{fig:AFexample3} as an example.
The initial variable equation set is $\{ A_1 \equiv \neg A_2, A_2 \equiv \neg A_1, A_3 \equiv \neg A_4, A_4 \equiv \neg A_3, A_5 \equiv \neg A_2 \neg A_4 \}$.
The variable equation for $A_1$ is already its refined equation, so we substitute that equation into the other equations.
After substitution, the equation for $A_2$ only changes to $A_2 \equiv \neg (\neg A_2)$.
Its equivalent equation is $A_2 \equiv A_2$, and the refined equation is $A_2 \equiv a_2$.
By a similar process for $A_3$ and $A_4$, the variable equation set becomes $\{ A_1 \equiv \neg a_2, A_2 \equiv a_2, A_3 \equiv \neg a_4, A_4 \equiv a_4, A_5 \equiv \neg a_2 \neg a_4 \}$.
From this set, we obtain the general allocator $E$ such that $E(1) = \neg a_2$, $E(2) = a_2$, $E(3) = \neg a_4$, $E(4) = a_4$ and $E(5) = \neg a_2 \neg a_4$.

We have implemented the algorithm as a prototype solver\footnote{\url{https://bitbucket.org/chiguri/allocation/}}.
Currently, the solver does not optimize some expressions about $U$, such as $U \lor a \neg a \equiv U$ and $U a \neg a \equiv a \neg a$.
As shown in Theorem~\ref{thrm:PNCMequation}, $U$ occurs in the refined equation, so the results of the solver are still redundant.
As the complexity of the algorithm depends on the size of the refined equations, redundant expressions impede the speed of the solver.
In contrast, the optimization for $U$ is not as simple as that for $T$ and $F$ in Lemma~\ref{lmm:expressionpattern}.
The analysis of the complexity is a future work.

\section{Local Allocation} \label{sec:localallocation}
As the size of a given framework is so big, we split it into ``blocks''.
Each block may be affected by other blocks.
The allocation method offers a method for abstracting such effects as logical variables.
We propose using the local allocation method reported in \cite{Moriguchi18}.
In this section, we use the term ``local allocation'' to refer to the method for blocks, and the term ``global allocation'' to refer to the method for the global framework (as in the above sections).

However, in \cite{Moriguchi18}, we only discussed the completeness of local allocation and acyclic cases for blocks.
Here, we introduce generality in a similar manner to global allocation.
Also, we propose a method to build general local allocators from a block using the equation-solving methods described in Section~\ref{sec:equation}.

\subsection{Block and Variable Argument}
First, we introduce the notion of blocks.

\begin{defn}[Block]
The triple $\langle \AR, \VArg, \ATT \rangle$ where $\AR$ and $\VArg$ are finite disjoint sets of arguments and $\ATT \subseteq (\AR \cup \VArg) \times \AR$ is called a block.
Each argument in $\AR$ is called the actual argument of the block, and each argument in $\VArg$ is called the variable argument of the block.
\end{defn}

A block is intended to denote a part of the framework.
For example, Figure~\ref{fig:AFexample5} shows the whole framework (the left-hand figure) and the block (the right-hand figure).
The block consists of the actual arguments $\{ 1, 2, 3 \}$ and the variable argument $a$.
The rest of the framework, arguments $\{ 4, 5, 6 \}$, constitutes the other block, which is isomorphic to the block described above (whose variable argument corresponds to $3$).
Note that the block can be seen as a part of the AF in Figure~\ref{fig:allocation_example}.
In this case, $a$ corresponds to $5$ in the framework.

\begin{figure}[tbp]
 \begin{center}
 \includegraphics[height=2cm]{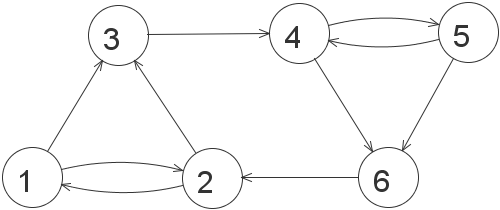} \qquad
 \includegraphics[height=2cm]{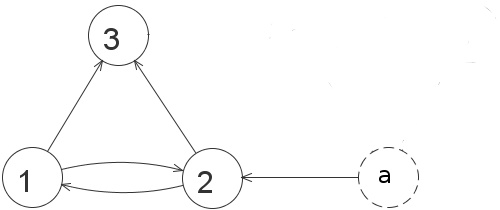}
 \end{center}
 \caption{An example of the framework and the block. The left-hand figure shows the original framework, and the right-hand shows the block.}
 \label{fig:AFexample5}
\end{figure}

A block in the framework depends only on actual arguments.
From this, we can build some blocks covering all of the framework by splitting a set of arguments into sets.
Formally, such ``splitting'' is defined as follows.

\begin{defn}[Splitter] \label{def:splitter}
Let $\langle \AR, \ATT \rangle$ be an argumentation framework.
A set of blocks $\mathbf{B}$ that satisfies the following conditions is called a splitter of the framework.
\begin{itemize}
 \item $\AR = \bigcup_{\langle \AR_i, \VArg_i, \ATT_i \rangle \in \mathbf{B}} \AR_i$, $\ATT =  \bigcup_{\langle \AR_i, \VArg_i, \ATT_i \rangle \in \mathbf{B}} \ATT_i$.
 \item For any distinct blocks $\langle \AR_1, \VArg_1, \ATT_1 \rangle, \langle \AR_2, \VArg_2, \ATT_2 \rangle \in \mathbf{B}$, $\AR_1 \cap \AR_2 = \emptyset$.
% \item For any block $\langle \AR', \VArg', \ATT' \rangle \in \mathbf{B}$, $\forall A' \in \AR'. \forall A \in \AR. (A, A') \in \ATT \Leftrightarrow A \in \AR' \cup \VArg' \land (A, A') \in \ATT'$. 実はこれARが被らない以上VArgが必要になることも必要なATTが全部入ってることも重複しないことも分かる。
\end{itemize}
\end{defn}

From the definition, no attacks are shared by the two distinct blocks.
Splitters can be used to build allocators from \emph{local allocators} (defined by the following) for the blocks.
\begin{defn}[Local Allocator]
Let $\langle \AR, \VArg, \ATT \rangle$ be a block.
A function from $\AR \cup \VArg$ to expressions is called a local allocator.
\end{defn}

The completeness of the local allocator is defined:
\begin{defn}[Complete Local Allocator]
A local allocator $E_l$ for the block $B$ is \emph{complete} iff the following conditions are satisfied.
\begin{itemize}
 \item For any variable argument $A$, $E_l(A) = a$ for some variable $a$ and $E_l(A) \neq E_l(A')$ for any other variable argument $A'$.
 \item For any actual argument $A$, $E_l(A) \equiv \bigwedge_{(A', A) \in \ATT} \neg E_l(A')$.
\end{itemize}
\end{defn}

In \cite{Moriguchi18}, we discussed the composition of local allocators and building of global allocators only in the case of an acyclic splitter.
In this section, we discuss the algorithm for building (local) allocators without such a limitation.

\subsection{Constantness and Generality of Local Allocation}
Here, we define constant local allocators and general local allocators.
As general global allocators are defined using constant global allocators, we will use constant local allocators in the definition of the generality of local allocation.
\begin{defn}[Constant Local Allocator]
A local allocator $E_l$ is constant iff for any argument $A$, $E_l(A)$ is a logical constant and satisfies $E_l(A) \equiv \bigwedge_{(A', A) \in \ATT} \neg E_l(A')$.
\end{defn}

Note that a constant local allocator is \emph{NOT} complete, in contrast to a constant global allocator, which is complete.
However, a constant local allocator is instantiated from a complete local allocator, as in Theorem~\ref{thm:instantiatevaluation}.

\begin{lmm}
For the complete local allocator $E_l$ and valuation $v$, $E'_l(A) = \evalV{E_l(A)}{v}$ is a constant local allocator.
\end{lmm}

Generality is also defined similarly in the case of global allocation.

\begin{defn}[General Local Allocator]
A complete local allocator $E_l$ is general iff for any constant allocator $E'_l$ there exists a valuation $v$ such that $E'_l(A) = \evalV{E_l(A)}{v}$.
\end{defn}

As we can see, when a block does not contain any variable arguments, it can be considered as an argumentation framework.
Moreover, in this case, complete/constant/general local allocators are complete/constant/general global allocators.

\begin{example} \label{example:cutblock}
Again, we use the right-hand figure in Figure~\ref{fig:AFexample5} as an example of the block.
It has three actual arguments $\{ 1, 2, 3 \}$, and its variable argument is $\{ a \}$.
As it is a complete local allocator, we can assign expressions to the arguments, such as $E_l(a) = a$, $E_l(1) = a$, $E_l(2) = \neg a$, and $E_l(3) = a \neg a$.
As constant local allocators, we can assign constants to the arguments, $T, F, U$ to $a$ and $1$, $F, T, U$ to $2$ and $F, F, U$ to $3$, respectively.
These constant local allocators are obtained from $E_l$, but there is also the other constant allocator, $F$ to $a$, $2$ and $4$, and $T$ to $1$, which is not obtained from $E_l$.
So, $E_l$ is complete, but not general.
On the other hand, a local allocator $E$ such that $E(a) = a$, $E(1) = a \lor b$, $E(2) = \neg a \neg b$ and $E(3) = (a \lor b) \neg a \neg b$, is general.

$E$ shows that the acceptability of $3$ depends on not only the external argument $a$ but also the internal argument $1$ (or $2$).
This means that this block requires an acceptance of $1$ (or a negativity for $2$) in the entire framework, including this block. %%% negativityのままでいいのかな
\end{example}

General local allocators represent the semantics of the block.
Fortunately, we can use an equation-solving method to construct general local allocators.

\begin{thrm} \label{thrm:local_general}
Let $\langle \AR, \VArg, \ATT \rangle$ be a block where $\AR = \{ A_1, A_2, \ldots A_n \}$ and $\VArg = \{ V_1, V_2, \ldots, V_n \}$.
A local allocator $E_l$ such that $E_l(A_i) = \RE{\mathbf{E}(A_i)}$ and $E_l(V_j) = V_j$ where $\mathbf{E} = \{ A_i \equiv \bigwedge_{(X, A_i) \in \ATT} \neg X | 1 \leq i \leq n \} [A_1] [A_2] \ldots [A_n]$, is a general local allocator for the block.
\end{thrm}
The proof is similar to that of Theorem~\ref{thrm:general_construction}.
The difference arises in the case of variable arguments, but it is easy to modify the proof because it does not require any conditions on variable arguments.

\subsection{Block Composition}
We obtain general local allocators for a block.
Our purpose is to compose some local allocators, but before discussing this, we define the composition of two blocks.

\begin{defn}
Let $B_1 = \langle \AR_1, \VArg_1, \ATT_1 \rangle$ and $B_2 = \langle \AR_2, \VArg_2, \ATT_2 \rangle$ be blocks satisfying $\AR_1 \cap \AR_2 = \emptyset$.
The block composed of these blocks is $\langle \AR, \VArg, \ATT \rangle$ where $\AR = \AR_1 \cup \AR_2$, $\VArg = (\VArg_1 \cup \VArg_2) \setminus \AR$ and $\ATT = \ATT_1 \cup \ATT_2$.
\end{defn}

Intuitively, the composition of two blocks instantiates variable arguments in one block to actual arguments in the other block.

Before discussing the composition of general local allocators, we present the properties of constant local allocators.
\begin{lmm} \label{lmm:compose_local_constant}
Assume that two blocks $B_1 = \langle \AR_1, \VArg_1, \ATT_1 \rangle$ and $B_2 = \langle \AR_2, \VArg_2, \ATT_2 \rangle$ satisfy $\AR_1 \cap \AR_2 = \emptyset$ and $B = \langle \AR, \VArg, \ATT \rangle$ is the composed block of $B_1$ and $B_2$.
Let $E_c$ be a constant local allocator for $B$; then $E_c$ is also a constant local allocator for $B_1$ and $B_2$.
\end{lmm}
\begin{proof}
From the definition of the composition of the blocks, for argument $A \in \AR_1$, $(A', A) \in \ATT \Leftrightarrow (A', A) \in \ATT_1$ for any $A'$.
This means that $E_c(A) \equiv \bigwedge_{(A', A) \in \ATT} \neg E_c(A') = \bigwedge_{(A', A) \in \ATT_1} \neg E_c(A')$, and therefore, $E_c$ is a constant local allocator for $B_1$.
It is similar for $B_2$.
\end{proof}

When composing two complete allocators for two blocks, two expressions are allocated into an instantiated argument; one is a variable (as a variable argument) and the other is an expression (as an actual argument).
The composition processes of each instantiated argument proceeds by making these expressions equivalent.
\begin{thrm} \label{thrm:compose_local_general}
Assume that two blocks $B_1 = \langle \AR_1, \VArg_1, \ATT_1 \rangle$ and $B_2 = \langle \AR_2, \VArg_2, \ATT_2 \rangle$ satisfy $\AR_1 \cap \AR_2 = \emptyset$ and $B = \langle \AR, \VArg, \ATT \rangle$ is the composition of $B_1$ and $B_2$.
Let $E_1$ be a general local allocator for $B_1$ and $E_2$ be a general local allocator for $B_2$, where $E_1(V) = E_2(V)$ for any $V \in \VArg_1 \cap \VArg_2$.
%%% ここ訂正漏れのまま出してしまった...
When $\mathbf{V} = \VArg_1 \cap \AR_2 \cup \VArg_2 \cap \AR_1 = \{ V_1, V_2, \ldots, V_n \}$ and
$\mathbf{E} = (\{ V \equiv E_2(V) | V \in \VArg_1 \cap \AR_2 \} \cup \{ V \equiv E_1(V) | V \in \VArg_2 \cap \AR_1 \}) [V_1] [V_2] \ldots [V_n]$,
$E$ defined as 
$$
E(X) = \left\{
\begin{array}{l@{\qquad}l}
E_1(X) [\mathbf{E}] & (X \in \AR_1 \cup \VArg_1 \land X \not\in \AR_2) \\
E_2(X) [\mathbf{E}] & (X \in \AR_2 \cup \VArg_2 \land X \not\in \AR_1) \\
\end{array}
\right .$$ 
is a general local allocator for $B$, where $p[\mathbf{E}] = p[p_1/x_1] \ldots [p_n/x_n]$ for $\mathbf{E} = \{ x_i \equiv p_i | 1 \leq i \leq n \}$.
\end{thrm}
\begin{proof}
First, we show that $E$ is complete.
From the definition of the composition of the blocks, for argument $A \in \AR_1$, $(A', A) \in \ATT \Leftrightarrow (A', A) \in \ATT_1$ for any $A'$.
This means that $E(A) = E_1(A)[\mathbf{E}] \equiv \bigwedge_{(A', A) \in \ATT_1} \neg E_1(A') [\mathbf{E}] = \bigwedge_{(A', A) \in \ATT} \neg E(A')$.
In a similar fashion, for argument $A \in \AR_2$, we can see that $E(A) = E_2(A)[\mathbf{E}] = \bigwedge_{(A', A) \in \ATT_2} \neg E_2(A') [\mathbf{E}] = \bigwedge_{(A', A) \in \ATT} \neg E(A')$.
Therefore, $E$ is complete.

Next, we show that $E$ is general.
Let $E_c$ be a constant local allocator for $B$.
From Lemma~\ref{lmm:compose_local_constant}, $E_c$ is a constant local allocator for both $B_1$ and $B_2$.
Because $E_1$ and $E_2$ are general local allocators for $B_1$ and $B_2$, respectively, there are valuations $v_1$ and $v_2$ where $\evalV{E_1(X)}{v_1} = E_c(X)$ for any $X \in \AR_1 \cup \VArg_1$ and $\evalV{E_2(X)}{v_2} = E_c(X)$ for any $X \in \AR_2 \cup \VArg_2$.
Here, we define $v$ as $v(X) = v_1(X)$ for any $X$ in the allocation variables of $E_1$ and $v(X) = v_2(X)$ for any $X$ in the allocation variables of $E_2$.
From the definition of $\mathbf{E}$, $\LV{\mathbf{E}} \cap \RE{\mathbf{E}} = \emptyset$, and there exists a valuation $v'$ such that $v'(X) = v(X)$ for any $X$, except for the extra variables introduced by solving the equation.
This means that $v'$ is a valuation satisfying $E_c(X) = \evalV{E}{v'}$, and therefore, $E$ is general.
\end{proof}

As we can see, the proof of the completeness of $E$ in the above theorem only depends on the completeness of $E_1$ and $E_2$.
This means that if $E_1$ and $E_2$ are complete, then $E$ is also complete.

\begin{coro} \label{coro:global_from_local}
For any splitter $\mathbf{B}$ of the framework $\langle \AR, \ATT \rangle$,
the allocator obtained by composition of general local allocators $E_i$ (for each block $B_i$ in $\mathbf{B}$) is a general allocator for $\langle \AR, \ATT \rangle$.
\end{coro}

%%% AF6の例を合成してAF5を？

From Theorem~\ref{thrm:local_general}, the general local allocators can be computed for each block in parallel.
Two independent compositions can be computed in parallel too, by Theorem~\ref{thrm:compose_local_general}.
Or, we can easily generalize the theorem, composing many general local allocators at once.
The local allocation method can be used to represent the semantics of the blocks and is also flexible about the computations of the framework.

Another application of the local allocation method is to present explicitly the effects between specific two arguments.
A global allocator denotes the relationship between an argument and the framework as its allocated expression.
We can compare two arguments via the relationship, i.e., expressions allocated to the arguments, however, it is difficult to clarify the effect from one to another.
In contrast, a general local allocator for the block with two variable arguments gives more explicit relationship between these arguments (corresponding to the variable arguments).
Here, we think about the splitter of the framework consisting two blocks, one is for only two arguments, namely $a$ and $b$, and the other is for the other arguments.
The latter block treats $a$ and $b$ as variable arguments, and its general local allocator $E$ uses the variables $a$ and $b$ as allocation variables.
The former block describes how $a$ and $b$ are attacked by the other arguments.
The conditions of completeness for $a$ and $b$ can be shown as $a \equiv \bigwedge_{(A, a) \in \ATT} \neg E(A)$ and $b \equiv \bigwedge_{(B, b) \in \ATT} \neg E(B)$.
In the refined equation of the former equation, $a$ occurs only left-hand side and $b$ occurs only right-hand side, and hence, it shows the effect from $b$ to $a$.
Similarly, the refined equation of the latter equation shows the effect from $a$ to $b$.
Such observation is possible by local allocation method.

\section{Discussion} \label{sec:discussion}
In this section, we show two aspects of general allocators; one is an application to stability, and the other is a new concept for AFs, termed arity.

\subsection{Stability} \label{sec:stability}
In Dung's AF, stable labelings are defined as complete labelings without any $\Lundec$ labelings.
They correspond to constant allocators without $U$.
A general allocator gives a constant allocator without $U$ when any expressions allocated by the general allocator are evaluated to $T$ or $F$ by valuation $v$.
From this observation, the conditions of the valuations are obtained as satisfiability problems.
We define a function $\mathcal{S}$ from constants in binary logic and three-valued logical expressions to binary logical expressions.
$$\begin{array}{lcl@{\qquad}lcl}
\mathcal{S}_T(T)            & = & T                                           & \mathcal{S}_F(T)            & = & F \\
\mathcal{S}_T(F)            & = & F                                           & \mathcal{S}_F(F)            & = & T \\
\mathcal{S}_T(U)            & = & F                                           & \mathcal{S}_F(U)            & = & F \\
\mathcal{S}_T(x)            & = & x                                           & \mathcal{S}_F(x)            & = & \neg x \\
\mathcal{S}_T(p_1 p_2)      & = & \mathcal{S}_T(p_1) \land \mathcal{S}_T(p_2) & \mathcal{S}_F(p_1 p_2)      & = & \mathcal{S}_F(p_1) \lor \mathcal{S}_F(p_2) \\
\mathcal{S}_T(p_1 \lor p_2) & = & \mathcal{S}_T(p_1) \lor \mathcal{S}_T(p_2)  & \mathcal{S}_F(p_1 \lor p_2) & = & \mathcal{S}_F(p_1) \land \mathcal{S}_F(p_2) \\
\mathcal{S}_T(\neg p)       & = & \mathcal{S}_F(p)                            & \mathcal{S}_F(\neg p)       & = & \mathcal{S}_T(p) \\
\end{array}$$
The functions $\mathcal{S}_T$ and $\mathcal{S}_F$ produce conditions to evaluate such an expression as $T$ and $F$, respectively.

\begin{lmm}
For expression $p$ and logical constant $C \in \{ T, F \}$, $\mathcal{S}_C(p)$ is satisfiable with a valuation $v$ iff $\evalV{p}{v} = C$.
\end{lmm}

From this lemma, we obtain a stable labeling by solving the satisfiability problem of $(\mathcal{S}_T(E(A_1)) \lor \mathcal{S}_F(E(A_1))) \land \ldots \land (\mathcal{S}_T(E(A_n)) \lor \mathcal{S}_F(E(A_n)))$ for $\AR = \{ A_1, \ldots, A_n \}$ and a general allocator $E$.

In this paper, we do not discuss the definition of stability of the allocation method.
We may define a stable allocator as a complete allocator without any occurrences of logical constants $U$, i.e., describable in binary logic, because we can obtain some stable labelings from such an allocator and valuations without $U$.
Fortunately, there are ``general'' stable allocators, with which any stable labelings can be instantiated, because any two constant allocators can be composed without $U$ by Theorem~\ref{thm:generalize}.
However, this definition is not applicable to extensions of AFs, such as abstract dialectical frameworks~\cite{Brewka17}.
A more general definition of stability in allocation methods will be investigated in future work.

\subsection{Arity}
We call the number of allocation variables for the general allocator \emph{the arity of the allocator}.
The arity of the general allocator built by the equation solving method described in Section~\ref{sec:equation} is the size of $\AR$.
When the expression $G$ is equivalent to $N \neg X \lor M$, where $N$ and $M$ satisfies the conditions in Lemma~\ref{lmm:PNCMform} (i.e., $P \equiv C \equiv F$), $\RE{\REF{X \equiv G}} \equiv U N \lor M$.
This means that some allocation variables introduced by the equation-solving method can be removed from the general allocator.

For any framework, the size of the allocation variables is less than the size of $\AR$ because the first attempt to substitute the equations satisfies the assumption described above\footnote{The equation is $A \equiv \bigwedge \neg A_i \land A$ or $A \equiv \bigwedge \neg A_i$ in the first attempt.}.
Considering the arity of a general allocator, the minimum arity might be one of the most important characteristics of a framework.
However, in our method, the processing order of equations affects the number of allocation variables.

\begin{example}
For the framework $\langle \{ 1, 2, 3 \}, \{ (1, 2), (2, 1), (1, 3), (3, 1) \} \rangle$ (Figure~\ref{fig:order_dependent}), if we solve the equations in the order $1$, $2$, and $3$,
 then the general allocator $E_1$ is $E_1(1) = \neg x_2 \lor x_3$, and $E_1(2) = E_1(3) = x_2 \neg x_3$.
However, if we solve the equations in the order $2$, $3$, and $1$, then the general allocator $E_2$ is $E_2(1) = x_1$, and $E_2(2) = E_2(3) = \neg x_1$.
\end{example}

\begin{figure}[tbp]
 \begin{center}
 \includegraphics[height=0.7cm]{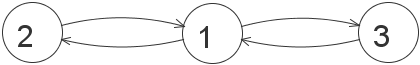}
 \end{center}
 \caption{A framework whose general allocators obtained by processing different orders of arguments have different arities.} \label{fig:order_dependent}
\end{figure}

An allocation variable is introduced when the equation has positive (or complex) self-dependencies.
Substitution of the refined equation moves the dependency of the argument (and its dependent arguments) to its attacking arguments. %%% これを図にするかで説明すべき
This means that self-dependency of an argument only occurs when the argument and some of the processed arguments make a loop.
At the same time, if an equation including self-dependencies only has negative occurrences, a variable is not introduced.
From this point of view, the minimum arity is, at most, the minimum size of the feedback vertex sets for the graph of the framework, and computing this is an NP-complete problem.

%%% なんかぶつっと切れてる感じがする

\section{Related work} \label{sec:relatedwork}

In the allocation method, we focused completeness (and stability).
So, we mainly compare with existing works for completeness.
In the future work, we will discuss other types of semantics proposed for argumentation frameworks.

\subsection{Extensions of Frameworks}
We have discussed allocation method for AFs, but the method is also applicable to extensions of AFs.
The following extensions allow for the flexibility of completeness, which controls each row of Table~\ref{tbl:example1}. %%% この書き方もっとやった方が良いのかな
The allocation method denotes the columns of the table and is orthogonal to these extensions.

Nielsen extends Dung's AF with set-attacks (SETAFs)~\cite{Nielsen07}.
In the extended framework, attack relations are defined as a subset of $2^{\AR} \times \AR$ instead of $\AR \times \AR$.
Basic notions of the semantics, such as completeness and stability, are shared with the original AFs.

Cayrol proposed bipolar argumentation frameworks (BAFs)~\cite{Cayrol13}.
BAFs involve not only attacks, but also supports.
Acceptabilities in BAFs reflect both adjacent arguments and indirectly supporting/attacking arguments.
This makes it hard to apply the allocation method to BAFs.
Kawasaki proposed another type of BAFs~\cite{Kawasaki18}.
Kawasaki's BAF also consists of supports and attacks, but the acceptabilities of the arguments are defined by adjacent arguments.

Abstract dialectical frameworks~\cite{Brewka17} (ADFs) are more expressive than AFs.
An ADF consists of a set of statements (corresponding to arguments in an AF), links between two statements, and acceptance conditions.
Acceptance conditions can be described as binary logical expressions with variables denoting the acceptabilities of the linked statements.
For statement $a$, the acceptance condition is the form $a \equiv \mathcal{F}(\mathit{par}(a))$ (here $\mathcal{F}$ is a function to a binary logical expression, and $\mathit{par}(a)$ is a set of variables denoting the acceptabilities of the linked statements).
This means that an acceptance condition in an ADF is also understood as a variable equation in Section~\ref{sec:equation}, and a set of conditions is a variable equation set.
Thus, it is possible to solve the variable equation set by the same algorithm as used for the AF, i.e., to construct a general allocator for an ADF.
The expressive power of ADFs enables us to describe SETAFs and Kawasaki's BAFs as ADFs; therefore, the allocation method covers these frameworks.

Our algorithm allows for a wider range of conditions than acceptance conditions in ADFs.
For example, if statement $a$ is attacked by $b$ and supported by $c$, its acceptance condition is denoted $a \equiv \neg b c$ or $a \equiv \neg b \lor c$.
However, we may want $a$ to be $\Lundec$ when both $b$ and $c$ are acceptable ($\Lin$).
In this case, we can describe the acceptance condition as $a \equiv \neg b c \lor U b c$ or $a \equiv (\neg b \lor c) (U \lor \neg b \lor \neg c)$.
Because the condition requires the logical constant $U$, it is impossible to represent this in an ADF, but easy to extend the definition of acceptance conditions.

Gabbay proposed semantics based on numerical equations~\cite{Gabbay12}.
In the equational approach, each argument has its own numerical function, describing relationships with other arguments.
The solutions of the equations are numeric values of $[0, 1]$ and interpretable as acceptance like labels ($0$ is `$\Lin$', and $1$ is `$\Lout$').
Such numerical equations are solved by general approaches with mathematical solvers such as Maple, MATLAB and NSolve.
However, the solutions of the equations are too detailed in terms of solutions (some solutions are denoted with square roots, which should correspond to `$\Lundec$'), and sometimes generate solutions outside of the range $[0, 1]$.
Moreover, the equational approach does not provide \emph{general} solutions, which we have concerned in this paper.

\subsection{Splitting Frameworks}

Splitting an AF into blocks and computing their semantics can be achieved by several methods.
Baroni et al. discussed the decomposability of a framework in several semantics~\cite{Baroni14}.
We do not discuss the types of splitters, but their results meet Corollary~\ref{coro:global_from_local}, i.e. any splitters are allowed in an allocation method based on complete semantics.
Variable arguments are similar to the input arguments in the $I/O$-gadgets, proposed in \cite{Giacomin16}.
Compared to $I/O$-gadgets, local allocation requires that we indicate all of the inputs (there are effects from outside of the block), but does not require any information about the outputs (effects on the outside).

Local allocation is also applicable to all of the extensions introduced in the previous section, in the same way as in the case of the AF.
Linsbichler proposed a splitting method for ADFs~\cite{Linsbichler14}.
This is a similar approach to that for AFs proposed by Baumann et al.~\cite{Baumann12}.
These approaches add some arguments like variable arguments in the local allocation, but they also add attacks (links in ADF).
Local allocation splits the framework with respect to the original topology, and it is therefore straightforward to split an ADF into two or more blocks.
On the other hand, allocation methods focus on completeness, but other semantics, such as preferred semantics or admissibility, are not discussed.

\subsection{Formula-based approaches}

There are many approaches to calculating several types of semantics using logical formulas.
These are based on extensions/labeling methods.

Besnard et al. proposed a methodology for encoding argumentation frameworks and set operations to logical formulas~\cite{Besnard14}.
Arieli and Caminada proposed an approach based on signed theories and quantified Boolean formulas to calculate several types of labelings~\cite{Arieli13}.
These approaches define a formula for each argumentation framework to represent the acceptance of arguments.
Our approach allocates a formula for each argument in a given framework to represent how the acceptability of the argument behaves in the framework.

The approach to the stability described in Section~\ref{sec:stability} is related to Besnard's approach.
Our approach uses two stages to generate formulas: calculating general allocators and generating formulas.
As general allocators are complete, the generated formulas do not explicitly include completeness conditions.
This means that the formulas generated by our approach may be smaller than those generated by Besnard's approach.
We should compare the sizes of the formulas generated by these approaches.

\section{Conclusion} \label{sec:conclusion}
In this paper, we proposed a new algorithm for constructing a general allocator.
In the algorithm, we solved the equations for three-valued logical expressions.
We demonstrated that the general solutions obtained by the algorithm are general allocators.
We also proposed an extension of the algorithm for constructing a general local allocator.
These algorithms are applicable to extensions of AFs such as SETAFs, BAFs and ADFs.

We will discuss the following three key areas in future work.
\begin{itemize}
 \item The complexity of the equation-solving algorithm:
 As the equation-solving algorithm proposed in this paper repeats the transformation of expressions, its time complexity depends on the sizes of the expressions.
 The sizes of the expressions increase mainly due to $\mathcal{R}$ in Definition~\ref{defn:transform_function} and substitutions of equations.
 We can decrease the sizes of the expressions by optimization, but some optimizations require specific patterns of expressions or canonicalization.
 
 \item Covering the dynamics of the framework:
 Addition/removal of arguments and/or attacks are often discussed in the context of belief revision.
 Because solving equations strongly relies on the topology of an AF, it is difficult to apply the algorithm to dynamics on an AF such as the addition/removal of arguments and/or attacks.
 Compositions of blocks are used to tackle such dynamics, for example, the division method proposed by Liao et al.~\cite{Liao11} and Booth et al.~\cite{Booth13}.
 A block introduced in Section~\ref{sec:localallocation} requires all attackers for its actual arguments, and therefore, it is still difficult to apply the addition of attacks from the outside of the block.
 We will develop more flexible structures for allocation method to apply dynamics.

 \item Discussing other types of semantics:
 We discussed some aspects of stability, but they are not applicable to ADFs.
 We would like to develop another definition of stability.
 To discuss wide types of semantics, we should develop admissibility in the allocation method.
 In the ADF, admissibility is defined using the partial order relation, which satisfies $U \leq T$ and $U \leq F$.
 Such relations are also definable in an allocation method similar to the equivalence relation $\equiv$, but it is necessary to discuss the existence of ``generally'' admissible allocators.
\end{itemize}

%\section*{Acknowledgment}
%This work was supported by JSPS KAKENHI Grant Number JP17H06103.

%%%%%%%%%%%%%%%%%%%%%%%%%%%%%%%%%%%%%%%%%%%%%%%%%%%%%%%%%%%%%%%%%%%

%%%%%%%%%%%%%%%%%%%%%%%%%%%%%%%%%%%%%%%%%%%%%%%%%%%%%%%%%%

\end{document}